\documentclass{article}

\usepackage{soul}   
\newcommand{\name}[1]{\texttt{\textbf{#1}}}
\usepackage{microtype}
\usepackage{graphicx}
\usepackage{subcaption}
\usepackage{booktabs} 

\usepackage{amsmath,amsfonts,bm}

\def\eqref#1{equation~\ref{#1}}

\def\1{\bm{1}}

\DeclareMathAlphabet{\mathsfit}{\encodingdefault}{\sfdefault}{m}{sl}
\SetMathAlphabet{\mathsfit}{bold}{\encodingdefault}{\sfdefault}{bx}{n}

\usepackage{makecell}
\usepackage{amssymb}
\usepackage{amsthm}
\usepackage[ruled,vlined,linesnumbered]{algorithm2e}
\usepackage{listings}
\usepackage{graphicx}
\usepackage{booktabs}
\usepackage{tabularx}
\usepackage{listings}
\usepackage{xcolor}
\usepackage{makecell}
\usepackage{booktabs}
\usepackage{tabularx}
\usepackage{algorithm2e}
\usepackage{array}
\usepackage{makecell}
\usepackage{algpseudocode}
\usepackage{booktabs}
\let\IF\If
\newcommand{\COMMENT}[1]{\tcp*[r]{#1}}
\newcommand{\STATE}{\relax} 
\newtheorem{theorem}{Theorem}[section]

\newtheorem{lemma}[theorem]{Lemma}

\newtheorem{definition}[theorem]{Definition}

\usepackage{enumitem}
\setlist[enumerate]{nosep, leftmargin=1cm}

\AtBeginDocument{%
  \setlength{\abovedisplayskip}{6pt plus 2pt minus 2pt}
  \setlength{\belowdisplayskip}{6pt plus 2pt minus 2pt}
  \setlength{\abovedisplayshortskip}{0pt plus 2pt}
  \setlength{\belowdisplayshortskip}{3pt plus 2pt minus 1pt}
}
 \usepackage[preprint]{neurips_2026}

\usepackage{quoting}

\usepackage[utf8]{inputenc} 
\usepackage[T1]{fontenc}    
\usepackage{hyperref}       
\usepackage{url}            
\usepackage{booktabs}       
\usepackage{amsfonts}       
\usepackage{nicefrac}       
\usepackage{microtype}      
\usepackage{xcolor}         

\title{AgentSociety: Incentivizing Agentic Social Intelligence}

\author{%
  Aditya Vema Reddy Kesari \thanks{Corresponding author: 22b3985@iitb.ac.in} \\
  IIT Bombay, India \\
  \And
  Krishna Reddy Kesari \thanks{Work does not relate to position at Amazon} \\
  Amazon, US \\
}

\begin{document}

\maketitle

\begin{abstract}
The success of deployed agents relies on their ability to handle open-ended user requests using their inherent capabilities, not only in solving requests directly but also in effectively leveraging inter-agent communication channels and feedback signals over time. This requires a multi-agent environment where agents can operate autonomously, strategically communicate, behave collaboratively and be driven by economic incentives, much like humans in society. Towards this vision, we propose \name{AgentSociety}, a mechanism that enables decentralized agentic collaboration grounded in liquid democracy and information diffusion from social choice theory. We show that \name{AgentSociety} provides an environment for agents to make autonomous decisions utilizing their local context to maximize their utility while achieving collective outcomes through incentivized collaboration. Specifically, we prove that delegation to more competent neighbor agents is incentive compatible and naturally generates multi-agent routing path by consensus. Additionally, our mechanism incentivizes agents to selectively disclose information to their neighbor agents when doing so aligns with their self-interest, so as to garner influence. We characterize the Nash equilibrium showing that agent payoffs are reflective of their marginal contributions. We compare and benchmark strategy profiles adopted by open and proprietary state-of-the-art language models deployed in \name{AgentSociety} against best response. Finally, we evaluate collaborative performance from consensus-based routing among self-interested heterogeneous agents in \name{AgentSociety} on real-world datasets.

\end{abstract}

\section{Introduction}
\label{intro}


As agents are increasingly deployed in real-world environments, they interact with humans and other agents over time. Endowed with non-trivial (reasoning) capabilities, their effectiveness emerges when they operate as utility-driven entities, requiring strategic behavior not only about \textit{what} to communicate, but also \textit{when}, \textit{why}, and with \textit{whom}, as information, objectives, and incentives evolve over time. However, existing characterizations of agent behavior largely focus either on single-agent evaluation through benchmarks \cite{mmlupro2, swe2, openlb2} or on (semi-) competitive agent interactions in game environments \cite{duan2024gtbench, huang2025competing, board-games, pokemon}. Such settings fail to reflect real-world environments, where agent success depends not only on independently solving requests, but also on leveraging inter-agent communication channels and prior feedback signals to complete requests beyond its standalone capabilities. Consequently, sustained test-time decentralized collaboration among heterogeneous deployed agents becomes essential. The challenge is to induce \textit{social} behavior, that is, mechanisms that incentivize self-interested agents to collaborate using only local information and contributing to collective outcomes while acting in their own interest. In such settings, communication is not merely an exchange of messages, but a strategic and utility-driven process. Therefore, we echo \cite{social-action}


\begin{list}{}{
    \leftmargin=0.5em 
    \rightmargin=0.5em
    \topsep=0pt      
    \partopsep=0pt 
    \parsep=0pt 
    \itemsep=0pt
}
\item  \textit{Agents are not agents by virtue of the fact that they communicate; they cannot be called social because they communicate but the other way around: they communicate because they are social}
\end{list}

A clear parallel can be drawn to humans in society, where effective collaborations are formed by operating autonomously and strategically with locally available information, typically achieving rewards proportional to incremental contributions. More concretely, society enables humans to operate in a \emph{social} manner, wherein optimizing for their own utility incentivizes strategic collaboration with their contacts, with collective outcomes largely being aligned to societal progress. To incentivize such behavior in agents, we propose \name{AgentSociety}, an incentive mechanism that enables autonomous social decisions by each agent. The autonomy is reflected in the agent's action space, as all actions are performed using local context and only if in its self-interest. \name{AgentSociety} aligns self-interested agent actions towards successful collective outcomes for user requests by associating key characteristics of agentic settings with mechanism design. \name{AgentSociety} is a decentralized, topography-agnostic framework leveraging incentivized hops for global knowledge and consensus while being inherently adaptive at runtime to evolving graph structures. By ensuring economically grounded, outcome-driven payoffs, it facilitates scalable autonomous collaboration and establishes a foundation for an open, heterogeneous, and fair agent economy \cite{vae, aex}.

Specifically, we leverage principles of diffusion auctions \cite{mdsn} and liquid democracy from social choice theory \cite{boldi, liqdem1, viscous2} to theoretically ground \name{AgentSociety}. The delegation objective for an agent is to either maximize its own probability of selection to perform a user request, conditioned on its intrinsic competence and local network topology, or failing that, to maximize its utility as a critical intermediary in the delegation path of a more suitable peer, as payoffs are contingent upon participating in the final successful routing path. The transitive delegation paths naturally provide feasible multi-agent routing paths, with the user request served by the path with highest accumulated votes. In order to garner votes of their neighbors, agents strategically and selectively diffuse information to their neighbors if in their self interest. The diffusion objective is formalized as a strategy of minimally sufficient information disclosure, signaling only the requisite threshold to achieve maximal expected utility. Overall, as the payoff is economically fair and conditional on the complete request being successful, agents across different tasks are incentivized to collaborate to obtain best possible (partial) view of the global state, thereby enabling decentralized multi-agent multi-task routing in heterogeneous LLM-based multi-agent systems (LaMAS), where constituent LLMs are powered by competing providers. Overall, our major contributions include -

\begin{enumerate}
    \item \textbf{Mechanism design for test-time collaborative LaMAS.} We present \name{AgentSociety}, a novel incentive mechanism for LaMAS grounded in auctions and social choice theory wherein self-interested heterogeneous agents are required to collaborate to maximize their utility. We show that \name{AgentSociety} involves an interplay of vote delegation with strategic self-interest driven peer information diffusion, governed by an overarching payoff design
    \item \textbf{Consensus based request routing.} We prove that delegation of an agent's vote to more competent neighbor is incentive compatible, thereby ensuring that consensus driven routing paths for user requests naturally arise through delegation chains within and across tasks, even in the presence of multiple heterogeneous agents with similar or overlapping capabilities, which is reflective of practical agentic scenarios and the society at large
    \item \textbf{Fair and marginal payoffs.} Agents receive payoffs that is fair and representative of their marginal contribution, conditioned on their strategic actions. At Nash equilibrium, we show that the payoff is representative of an agent's true marginal contribution over the system
    \item \textbf{Social intelligence benchmarking.} We evaluate the social reasoning of several open and proprietary LLMs deployed in \name{AgentSociety} against best response on competence reporting, collaborative information diffusion and incentive compatible delegation dimensions 
    \item \textbf{Collaborative performance on real world datasets.} We evaluate collaborative performance of heterogeneous self-interested agents within \name{AgentSociety} on MMLU-Pro \cite{mmlupro2}, Open LeaderBoard v2 \cite{openlb2} and SWE-bench \cite{swe2}

\end{enumerate}


The remainder of the paper is organized as follows. \S~\ref{rw} discusses related work, \S~\ref{setup} introduces the setup of our multi-agent system and \S~\ref{mechanism} presents our mechanism. \S~\ref{dynamics} discusses collaborative dynamics induced by the mechanisms with empirical evidence presented in \S~\ref{results} before concluding.

\section{Related Work}
\label{rw}
Recent literature has increasingly advocated for framing multi-agent systems through the lens of social choice theory and market design. However, these works primarily establish theoretical desiderata rather than providing mechanisms or implementations. For instance, the position paper MASS \cite{mass} discusses the necessity for multi-agent AI systems to incorporate social theory as a structural prior. Similarly, Virtual Agent Economies \cite{vae} emphasizes the necessity of proactively engineering steerable agent markets to ensure alignment with collective flourishing, while Agent Exchange \cite{aex} outlines the high-level system challenges in adapting real-time bidding architectures to agentic ecosystems. While these conceptual frameworks successfully identify critical systemic requirements, they fall short of providing executable mechanisms to realize them. In contrast, our proposed framework \name{AgentSociety} directly bridges this gap by operationalizing these high-level desiderata into a functional mechanism. Specifically, our system enables incentivized collaboration among self-interested agents through a novel mapping of individual utilities to the achievement of collective goals. In doing so, we also address a key limitation of prior cooperative multi-agent formulations in their reliance on Dec-POMDPs \cite{decpomdp1, fermat, decpomdp2}, where agents optimize a shared global reward. While POMDPs are effective in tightly cooperative settings, it abstracts away heterogeneous incentives and localized objectives that are central to realistic multi-agent deployments. Therefore, \name{AgentSociety} instead models interactions as a Partially Observable Stochastic Game (POSG), in which each agent exhibits an individual reward function shaped by its capabilities, neighborhood, and position in the interaction graph. This formulation more naturally captures decentralized socio-economic behavior, but also makes explicit mechanism design necessary to align individual incentives with system-level outcomes. To this end, \name{AgentSociety} encompasses a novel combination of diffusion auctions \cite{idm} with liquid democracy \cite{liqdem1, liqdem2}. Auctions provide an economically grounded approach to resource allocation and compensation, while liquid democracy enables decentralized consensus and routing path generation through transitive delegation.

This test-time economic incentive-driven collaboration differentiates our approach from existing LLM-based multi-agent systems which primarily learn communication protocols during training \cite{gptswarm, mas-gpt, supernet, gdesigner, aflow, agentnet} and emphasize token efficiency \cite{cutthecrap}. These approaches do not account for the economic structure of interaction, leaving compensation across heterogeneous agents and chraacterization of incentive-driven emergent behavior unresolved. In contrast, our framework explicitly models payoffs representative of marginal contribution, enabling self-interested heterogeneous agents, potentially from different LLM providers, to collaborate dynamically at test time without prior training. Critically, \name{AgentSociety} provides a measure the reasoning ability of LLM-agents in social and collaborative environments, in contrast to prior benchmarks such as GTBench \cite{duan2024gtbench} and GAMA-Bench which \cite{huang2025competing} evaluate strategic reasoning in (semi-)competitive environments \cite{board-games, pokemon}. While GOVSIM \cite{govsim} explores cooperation through open-ended negotiation and reasoning, it relies on an agent's internal ability to simulate long-term consequences. Consequently, while existing work characterizes strategic ability for individual gain, we frame it as the foundational social intelligence characteristic in designing mechanisms for decentralized agent coordination, moving beyond games towards designing mechanisms for society-scale agentic collaborative and consensus based request orchestration.

\section{Setup}
\label{setup}
We model a multi-agent structure as an undirected graph $G = (\mathcal{N}, E)$, where $\mathcal{N}$ is a set of agents and $E$ represents communication links between them. An agent $n_i \in \mathcal{N}$ is defined by four components: $n_i = \{ \textit{base}_i, \textit{state}_i, \textit{tools}_i, \textit{role}_i \}$, where $\textit{base}_i$: the base language model; cognitive core of the agent, $\textit{state}_i$: the agent's internal state representing its memory and context, $\textit{tools}_i$: the set of external tools the agent can access, $\textit{role}_i$: the agent's specialized role and permissions within the system. Graph $G$ has a societal-inspired community structure, with agents belonging to one of $k$ communities $C_1, C_2, \dots, C_k$, wherein belonging to the same community represents similar capabilities in terms of role and tools to solve task $t_k$ albeit with varying competence, such that
$
\mathcal{N} = \bigcup_{i=1}^k C_i, \quad C_i \cap C_j = \emptyset \text{ for } i \ne j.
$
The edge set $E$ is partitioned into:
$
E = E_{\text{intra}} \cup E_{\text{inter}},
$
where $e_{i,j} \in E_{\text{intra}}$ if agents $n_i$ and $n_j$ belong to the same community and $e_{i,j} \in E_{\text{inter}}$ if they belong to different communities.

\begin{definition}[Partially Observable Stochastic Game]
A \emph{POSG} is a tuple $
\mathcal{G} =
\langle
\mathcal{N},
\mathcal{S},
\{\mathcal{A}_i\}_{i \in \mathcal{N}},
T,
\{P_i\}_{i \in \mathcal{N}},
\{\mathcal{O}_i\}_{i \in \mathcal{N}},
O,
\gamma
\rangle
$, where $\mathcal{N}$ is the set of agents. At each step, given state $s \in \mathcal{S}$ and joint action $\mathbf{a} \in \mathcal{A} = \prod_i \mathcal{A}_i$, the environment transitions via kernel $T(s' \mid s, \mathbf{a})$. Agents receive private observations $\mathbf{o} \in \prod_i \mathcal{O}_i$ according to $O(\mathbf{o} \mid s', \mathbf{a})$ and individual payoffs according to $P_i(s, \mathbf{a})$.
\end{definition}

\begin{figure*}[t]

\centering
\includegraphics[width=\textwidth]{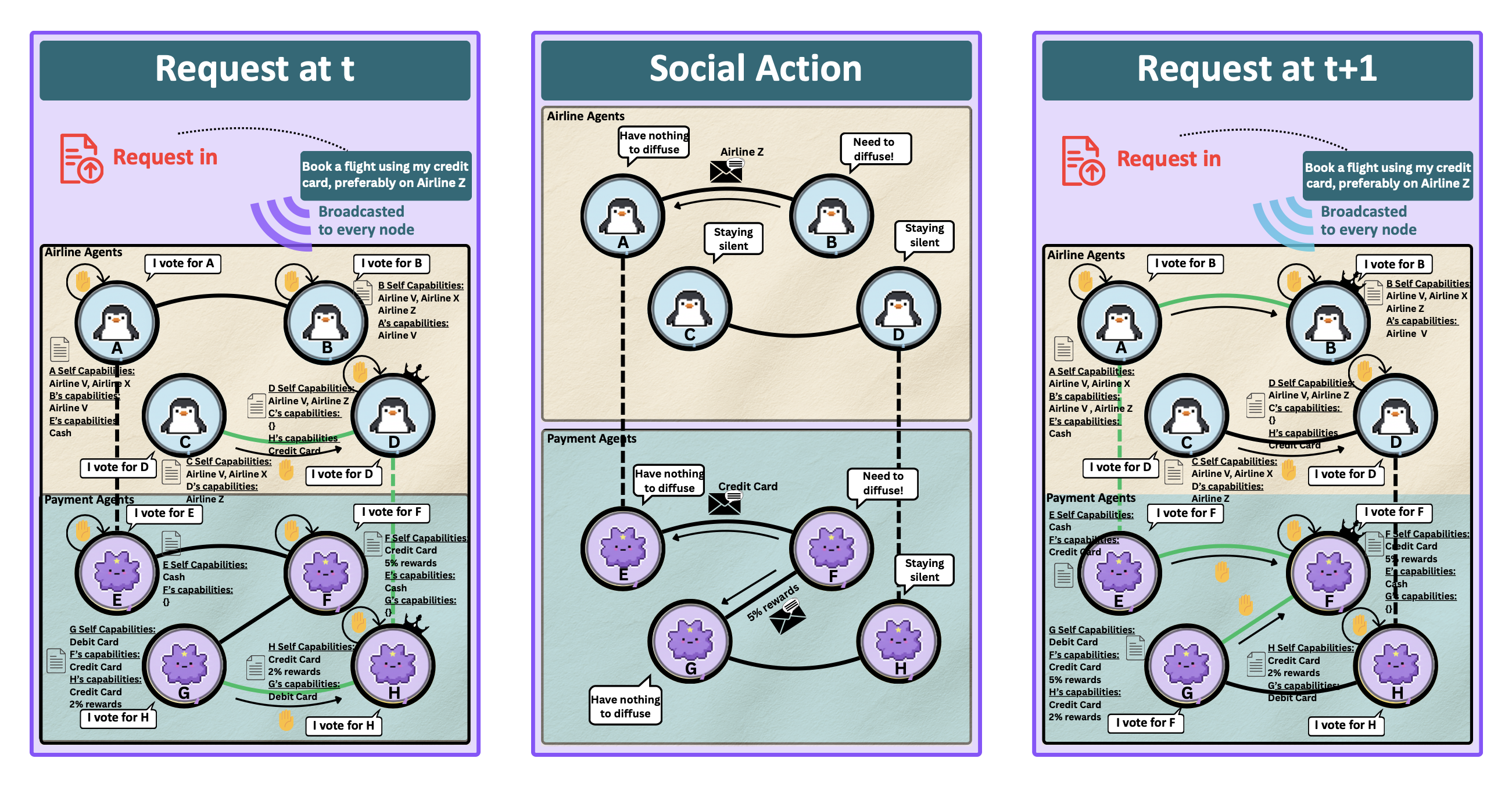}
\caption{The user request reaches all agents in \name{AgentSociety} via broadcast and qualified agents participate. For each task, agents having similar capabilities arrive at task allocation by consensus as depicted to the \emph{left}. The winning path C-> D -> H -> G (denoted in green) with maximum votes performs the user request with delivered competence being Airline D booking with 2\% rewards. As depicted in the \emph{center}, agents then perform social action in their self interest. Since Agents B and F did not perform the task when they were competent, it is in their self-interest to selectively diffuse information to A and E, G respectively to garner their influence. Since we prove incentive compatibility in delegation (Theorem \ref{thm:ic}), voting for a more competent neighbor is in the self interest of A, E and G. As a result of this social action, any similar future request would now be served by the winning path (higher votes) B -> A -> E -> F -> G (denoted in green) on the \emph{right} delivering the user request with increased competence of Airline D with 5\% rewards, while being a win-win-win for the user, agent performing the task as well as the intermediary agents.}
\label{fig:overall}
\end{figure*}

User requests in \name{AgentSociety} reach all agents $n_i \in \mathcal{N}$ via broadcast where splitting agent ($\Omega$) handles request $Q$ dissemination $
\Omega(Q) = [t_{1}, t_{2}, \ldots, t_{m}]
$ as an ordered sequence of tasks.Each agent evaluates and performs self-interested strategic delegation and diffusion actions. Delegation involves transferring its assigned vote to another agent for the current request, while information diffusion is governed by the observed payoff for the request, with agents selectively disclosing information about their capabilities and competence when doing so is expected to increase future payoffs. Diffusion can occur across any edge $e_{i,j} \in E$ while delegation is restricted to edges $e_{i,j} \in E_{\text{intra}}$.

\section{\name{AgentSociety} Mechanism}
\label{mechanism}


The objective of our mechanism is to facilitate decentralized collaboration among self-interested agents to optimize user request execution driven by consensus. Upon receiving query $\Omega(Q) = [t_1, t_2, \dots, t_m]$, each agent $n_i \in \mathcal{N}$ is eligible to participate in task $t_j \in Q$ if and only if $n_i \in \mathcal{T}_{t_j}$, where $\mathcal{T}_{t_j} \subseteq \mathcal{N}$ denotes the subset of agents capable of executing $t_j$. The mechanism is structured across four key components: the state space encompassing competence on capabilities and interaction history, agent delegation and competence reporting to a ledger, computation of delegation-based routing path and payoffs, resulting payoff observations by agents and social diffusion action through which agents garner influence to affect future delegation, as illustrated in Fig. \ref{fig:overall}. 


\subsection{State Space}
The system state $s^\tau \in \mathcal{S}$ at step $\tau$ is: $s^\tau = \bigl(\{\mathbf{c}_i\}_{n_i \in \mathcal{N}},\; \{\hat{\mathbf{c}}_{j \to i}^{<\tau}\}_{n_i \in \mathcal{N},\, j \in r_i},\; \{\Pi_i^{<\tau}\}_{n_i \in \mathcal{N}} \bigr) $, where $r_i$ denotes immediate neighbors of agent $n_i$ and we elaborate on each term. First, the vectors $\{\mathbf{c}_i\}_{n_i \in \mathcal{N}}$ are the intrinsic competences of all agents across $k$ task classes; these are private to each agent and never directly revealed to the ledger or to peers. Second, $\hat{\mathbf{c}}_{j \to i}^{<\tau}$ is the competence signal agent $n_j$ has diffused to neighbor $n_i$ as a part of the overall prior information diffused by agent $n_j$ given by $h_j^{<\tau}$ (elaborated further in \S~\ref{sec:diffusion}). This diffused competence signal $\hat{\mathbf{c}}_{j \to i}^{<\tau}$ determines $n_i$'s current belief about $n_j$'s competence and directly shapes $n_i$'s delegation decisions. Third, $\Pi_i^{<\tau} = \{p_i^{\tau'}\}_{\tau' < \tau}$ is the history of payoffs realized by agent $n_i$, which encodes the cumulative feedback on the quality of $n_i$'s prior delegation, competence reporting and diffusion actions. The state $s^\tau$ is never fully observed by any single agent. Each agent $n_i$ observes its intrinsic competence $\mathbf{c}_i$ and its own payoff history $\Pi_i^{<\tau}$, but can only access the competence of neighbor $n_j$ through the signals $\hat{\mathbf{C}}_{j \to i}^{<\tau}$ that $n_j$ has chosen to diffuse, signals that are, in general, strategically compressed below $n_j$'s intrinsic capability (described further in \S~\ref{sec:diffusion}). The joint state transition $s^{\tau+1} \sim T(\cdot \mid s^\tau, \mathbf{a}^\tau)$ is driven both by the agents' strategic actions, with each agent seeking to maximize its cumulative payoff. 

\subsection{Agent Delegation Action and Strategic Competence Reporting}

On receiving a user request $Q_i$, each capable agent $n_i \in \mathcal{T}_{t_k}$ submits a strategic delegation action tuple based completely on its local information $a_{i}^\tau (t_k, \{\hat{\mathbf{c}}_{j \to i}^{<\tau}\}_{n_i \in \mathcal{N},\, j \in r_i}, \Pi_i^{<\tau}) = (\mathbf{c}'_i,\, v_i,\, r_i)$ to a ledger for consensus-based routing and payoff computation for request $Q_i$. The \textit{reported competence vector} $\mathbf{c}'_i$ is the agent's declared task capability, which may strategically deviate from its intrinsic competence $\mathbf{c}_i$. The \textit{delegation decision} $v_i \in r_i \cup n_i$ determines whether agent $n_i$ retains its single allocated vote or transfers it to a neighbor. Votes follow transitive delegation: delegating transfers both the agent's own vote and all votes it has previously accumulated, thereby accumulating along delegation chains. The strategic character of $(\mathbf{c}'_i, v_i, r_i)$ arises from the information asymmetry encoded in the state space: because intrinsic competences $\mathbf{c}_i$ are visible only to the agent themselves and the delegation decision $v_i$ of agent $n_i$ is made on the basis of the diffused signals $\hat{\mathbf{c}}_{j \to i}^{<\tau}$ received from neighbors rather than ground-truth intrinsic competence $c_j$ of neighbor $n_j$. Therefore, the reported $\mathbf{c}'_i$ is chosen to maximize expected payoff given these, further elaborated in \S~\ref{dynamics}. To ensure aligning global goals to individual incentives, we prove incentive compatibility in delegation to more competent neighbors, that is, an agent $n_i$ always gains by delegating to a more competent neighbor $n_j$ as assessed from its local information that encompasses the information diffused by agent $n_j$. 

\begin{theorem}[Incentive Compatibility in Delegation]
\label{thm:ic}
In a rational system, for any agent $n_i \in \mathcal{N}$, suppose there exists a neighbor $n_j \in r_i^\tau$ such that $\hat{c}_{j \to i, t_k}^{\tau -1} > c_{i,t_k}$. Then, for any joint action profile $a_{-i}^\tau$ of the remaining agents, the optimal delegation choice of agent $n_i$ satisfies $v_i^\tau \neq n_i$. More precisely, letting $a_i^{j, \tau}(t_k, \{\hat{\mathbf{c}}_{j \to i}^{i<\tau}\}_{n_i \in \mathcal{N},\, j \in r_i}, \Pi_i^{<\tau}) = (\mathbf{c}'^k_i, v_i = n_j,\, r_i)$ and $a_i^{i, \tau}(t_k, \{\hat{\mathbf{c}}_{j \to i}^{i<\tau}\}_{n_i \in \mathcal{N},\, j \in r_i}, \Pi_i^{<\tau}) = (\mathbf{c}'^k_i, v_i = n_i,\, r_i)$, with $u_i$ denoting the utility of agent $n_i$, we have:
\begin{equation*}
u_i\big( a_i^{j, \tau},\, a_{-i}^\tau \big) \geq u_i\big( a_i^{i, \tau},\, a_{-i}^\tau \big).
\end{equation*}
\end{theorem}

We refer the reader to  \S~\ref{ic} for the complete proof. 


\subsection{Consensus based routing paths and critical nodes}
\label{sec:ledger}

All self-interested decisions $\{(\mathbf{c}'_i, v_i, r_i)\}_{n_i \in \mathcal{N}}$ from agents are aggregated to produce the routing path and compute agent payments, by a ledger that is a transparent, non-coercive aggregator and whose operation is common knowledge to all agents.

\paragraph{Consensus Based Routing paths.}
An agent that retains its vote is called a \textit{guru} $g \in \mathcal{G}_k$ for task $t_k$. The set $\mathcal{D}(g)$ denotes all agents whose votes transitively flow to $g$, with total vote count $V(g) = |\mathcal{D}(g)| + 1$, and this delegation path forms the routing path for a single task. To extend a routing path from task $t_k$ to task $t_{k+1}$, the ledger identifies the \textit{representative delegate}
$
    n_d^* = \arg\max_{n_d \in \mathcal{D}(g)} \sum_{j > k} c_d^{\prime\,t_j},
$
the member of $g$'s delegation pool with the greatest reported cumulative competence over downstream tasks. If $n_d^*$ has declared a connection to some $n_u \in \mathcal{T}_{t_{k+1}}$, the path is extended to the corresponding guru $g' \in \mathcal{G}_{k+1}$ of $n_u$, with $(n_d^*, n_u)$ serving as a bridge across task classes, and the path vote count updated as $V_{\text{path}} \leftarrow V_{\text{path}} + V(g')$. This procedure is applied iteratively across $Q$. A path is \textit{feasible} if it maintains a continuous chain of connected gurus across all tasks, and the ledger selects
$
    \mathcal{P}^* = \arg\max_{\mathcal{P} \in \mathbb{F}}\; V_{\text{path}}
$
over the feasible set $\mathbb{F}$.

\paragraph{Critical Agents and Payoff Computation.}
To compute payments reflecting marginal contribution, we identify agents whose delegation decisions are indispensable to $\mathcal{P}^*$. For any $n_i \in \mathcal{P}^*$, consider the counterfactual in which $n_i$ retains its vote, and let $\mathcal{P}_c^*(n_i)$ be the resulting feasible path. Agent $n_i$ is \textit{critical} if $\mathcal{P}_c^*(n_i)$ is the counterfactual winning path. The critical agents
$
    \mathcal{C} = \bigl\{\, n_i \in \mathcal{P}^* \mid \mathcal{P}_c^*(n_i) \text{ is the winning path}\bigr\}
$ is the set of all such agents $n_i$.

\begin{theorem}[\textbf{Critical Chain}]
\label{thm:critical}
For each task $t_k$, the critical set $\mathcal{C}_{t_k} = \{n_i \in \mathcal{C} \mid n_i \in \mathcal{T}_{t_k}\}$ forms a contiguous delegation chain in which each agent delegates its vote to the next agent in the sequence.
\end{theorem}

\noindent We refer the reader to \S~\ref{proof-critical} for the proof. 

Let $\mathcal{C}_{t_k} = (n_1, n_2, \dots, n_{m-1}, n_m = g)$ be the ordered critical chain. The payoff to $n_i \in \mathcal{C}_{t_k}$ is:
\begin{equation}
\label{eq:del_payoff}
p^{(t_k)}(n_i) =
\begin{cases}
\mathbb{I}_{\{c'_{i+1} > c'_i\}}\bigl(f(c'_i) - f(c'_{i-1})\bigr), & n_i \in \mathcal{C}_{t_k} \setminus \{g\}, \\[6pt]
c + \mathbb{I}_{\{c'_{i+1} > c'_i\}}\bigl(f(c'_i) - f(c'_{i-1})\bigr), & n_i = g, \\[6pt]
0, & n_i \notin \mathcal{C}_{t_k},
\end{cases}
\end{equation}
where $f(\cdot)$ is a monotone scaling function, $c$ is a base execution cost, and $\mathbb{I}$ equals $-1$ if $c'_{i+1} > c'_i$ and $\alpha > 0$ otherwise, thereby penalizing irrational delegation to a less competent agent on task $t_k$. The aggregate property $\sum_{n_i \in \mathcal{C}_{t_k}} p^{(t_k)}(n_i) = c + f(c_g)$ pins the user's total cost for task $t_k$ to a function of the competence of the executing guru, which implies the user's cost to be proportional to competence delivered.
The \textit{infeasibility penalty}
$
    p^{\text{inf}}(n_i) = \alpha \sum_{t \neq t_k} c_i^{\prime\,t}
$
is levied when an agent declares an infeasible connection, so as to reliably extend delegation paths across tasks. The \textit{auxiliary misreporting penalty}
$
    p^{\text{mis}}(n_i) = \alpha \sum_{n_{i'} \in \mathcal{D}_g^{t_k}} \min\!\bigl(0,\; c_{i'}^{\prime\,\neg t_k} - c_i^{\prime\,\neg t_k}\bigr)
$
penalizes misreporting of competence to the ledger, where $c_{i}^{\prime\,\neg t_k}$ represents the competence reported by $n_i$ on tasks other than $t_k$ and $c_{i'}^{\prime\,\neg t_k}$ represents the competence achieved on tasks other than $t_k$. Importantly, we distinguish between permitted strategic signaling and penalized misreporting based on intrinsic capacity and prior information flow. A report $\hat{c}_i > c_i$ is considered strategic and non-penalized if a delegated neighbor $n_j$ possesses sufficient competence ($c_j \geq \hat{c}_i$) and has previously diffused a claim $\hat{c}_j \geq \hat{c}_i$. This allows agents to safely leverage the verified potential of their neighborhood. Conversely, reporting a competence that exceeds both intrinsic capacity and the previously diffused capabilities of all neighbors is classified as misreporting. Critically, this can be verified by the ledger through the competence reports while determining the payoff. Therefore, the total payoff is given by:
$
    p^{\text{total}}(n_i) = \sum_{t_k \in Q} p^{(t_k)}(n_i) \;-\; p^{\text{mis}}(n_i) \;-\; p^{\text{inf}}(n_i).
$

\subsection{Agent Observations and Diffusion Actions}
\label{sec:diffusion}

\paragraph{Payoff.}
Once the ledger resolves $\mathcal{P}^*$ and computes payments, each agent $n_i$ receives a private observation:
$    o_i^{\tau+1} = \Bigl(p_i^\tau,\;\; \bigl\{\hat{\mathbf{c}}_{j \to i}^\tau,\; m_j^i\bigr\}_{j \in r_i}\Bigr),
$
comprising two components. The first is a \textit{realized payoff} $p_i^\tau$ is feedback on the effectiveness of agent $n_i$'s strategic reports $(\mathbf{c}'_i, v_i, r_i)$ in the resolved allocation for the request: it signals whether $n_i$'s delegation decision was effective in securing a position in $\mathcal{P}^*$, updating the payoff history $\Pi_i^\tau$ in its state. The second is information received through diffusion actions by neighbors elaborated below. This peer information exchange is performed by each agent when in its self interest and potentially conditioned on the payoff observed. This information diffused influences future delegation outcomes.  

\paragraph{Diffusion Action.}
Each agent $n_j$ diffuses information conditioned on the payoff:
$
    h_j^\tau (\{\Pi_j^{<\tau}\} \bigr) = \bigl\{\,(j \to i,\; \hat{\mathbf{c}}_{j \to i}^\tau,\; m_j^i)\bigr\}_{i \in r_j},
$
where $\hat{\mathbf{c}}_{j \to i}^\tau$ is the competence signal agent $n_j$ transmits to neighbor $n_i$ and $m_j^i \in \mathcal{M}$ is an optional accompanying private message, typically conditioned on historical payoffs received. The \textit{incoming diffusion signals} to agent $n_i$ given by $\{\hat{\mathbf{c}}_{j \to i}^\tau, m_j^i\}_{j \in r_i}$ are the outputs of each neighbor $n_j$'s social action $h_j^\tau$, and constitute the \textit{sole channel} through which agent $n_i$ can form beliefs about the capabilities of its neighbors. Unlike the ledger report $\mathbf{c}'_i$, the diffused signal $\hat{\mathbf{c}}_{j \to i}^\tau$ carries no direct payment consequence. Its purpose is entirely prospective: by shaping what neighbor $n_i$ observes, agent $n_j$ influences $n_i$'s future delegation decision $v_i$ and therefore its own future vote accumulation and payoff. The central tension in the diffusion channel is between influence and exposure. An agent $n_j$ benefits from attracting delegation from neighbor $n_i$, but fully revealing its true competence $\mathbf{c}_j$ allows $n_i$ to form accurate beliefs that diminish $n_j$'s payment in future rounds given by Eq. \ref{eq:del_payoff}, elaborated further in \S~\ref{dynamics}. We characterize the Nash equilibrium across the delegation, diffusion and strategic reporting dimensions, wherein each agent receives payoff representative of their true marginal over the system.

\noindent

\begin{theorem}[Nash Equilibrium Reporting, Diffusion, and Delegation] 
At Nash equilibrium action profile $a = (a_i)_{n_i \in \mathcal{N}}$, the following conditions hold for every agent $n_i \in \mathcal{N}$:
\vspace{-5pt}
\begin{center}
\small
\begin{tabular}{ccc}
\textbf{Reporting:} & \textbf{Diffusion:} & \textbf{Delegation:} \\
$c'^i_{t_i} = \begin{cases} c^i_{t_i}, & v_i = n_i \\ \hat c_i, & v_i \neq n_i \end{cases}$ & 
$c_{t_i}^{j \to i} = \begin{cases} c'^i_{t_i}, & c'^j_{t_i} > c'^i_{t_i} \\ c'^j_{t_i}, & c'^j_{t_i} \le c'^i_{t_i} \end{cases}$ & 
$v_i = \begin{cases} i, & c_i \ge \tilde d_i \\ \text{arg}\max_{j \in r_i} c_{t_i}^{j \to i}, & c_i < \tilde d_i \end{cases}$
\end{tabular}
\end{center}
\vspace{-5pt}
\label{thm:ne}
\end{theorem}
\noindent where $\hat c_i
=
\max_{v \in S(i,v_i)\setminus
\arg\max_{u \in S(i,v_i)} c_u} c_v$, 
$\tilde d_i = \max_{j \in r_i} c_{t_i}^{j \to i}$ and
$S(i,j) = \{ n_k \in \mathcal{N} \;|\; n_j \in \operatorname{reach}(n_k) \;\text{and}\; n_k \in P,\; \forall P \in \mathcal{P}_{j \to i} \}$,
where $\operatorname{reach}(n_i) = \{ n_j \in \mathcal{N} \mid \exists \text{ a directed path from } n_i \text{ to } n_j \}$ denotes the set of agents reachable from $n_i$, and $\mathcal{P}_{j \to i}$ denotes the set of all directed paths from $n_j$ to $n_i$ in the diffusion graph. We refer the reader to \S~\ref{ne} for the complete proof.

We provide the the overall \name{AgentSociety} pseudo-code in \S~\ref{pseudo}.

\section{Mechanism-Induced Collaborative Dynamics}
\label{dynamics}


\paragraph{Interplay of Diffused Competence and Reported Competence} Agents optimize a neighbor-specific disclosure to signal competence while minimizing opportunity costs via private communication $h_j^\tau \in \mathcal{H}$. The strategy therefore is: agent $n_j$ diffuses the minimum signal sufficient to secure $n_i$'s delegation given $n_i$'s current beliefs about the competitive landscape. Formally, the equilibrium diffusion satisfies:
$
    \hat{c}_{j \to i}^\tau = c_i + \varepsilon,
$
strictly dominating $n_i$'s self-reported competence by the smallest margin $\varepsilon > 0$ that triggers delegation, rather than disclosing $c_j$. Agent $n_i$, observing a superior signal from $n_j$, delegates to $n_j$ given incentive compatibility and updates its ledger report accordingly to be $c_i + \epsilon$; while agent $n_j$ reports its intrinsic competence $c_j$ to the ledger and captures the payment differential described in Eq. \ref{eq:del_payoff}. 
Therefore, agent $n_i$'s beliefs about neighbor capabilities are strategically downward-biased unless required to be higher: the accumulated history $\hat{\mathbf{C}}_{j \to i}^{<\tau}$ in the state reflects not intrinsic competences but the sequence of influence bids that neighbors have found it strategically optimal to reveal. This selective signaling mechanism inherently induces: (i) \textbf{Dynamic Calibration}, where agents incrementally increase disclosure intensity only upon failure to attract delegation; (ii) \textbf{Persistent Interaction}, where stable signaling channels transform private informational advantages into long-term utility; and (iii) \textbf{Informed Delegation}, where private states provide the requisite local context for consensus-based routing. The net effect is the system discovers incentivized routing paths that provide higher performance to the user request and the intermediaries are rewarded by their ability to declare higher than their intrinsic capability and is not considered misreporting as this increased competence can be achieved by delegating its vote to the more competent neighbor, making it a win-win-win for the intermediary, more competent neighbor and the user.

\paragraph{Incentivized Transitive Delegation and Consensus Routing}
The mechanism moves from routing by competition to collaboration by inducing a state of transitive vote augmentation, where agents maximize utility by delegating to specialized peers to incrementally build a path's aggregate competence. Rather than engaging in a contest for singular task allocation, agents with overlapping expertise find augmentation to be the dominant strategy; by contributing to a high-performance task routing that naturally arises from vote delegation paths, agents secure a share of a larger, more certain payment pool that isolated bidding cannot provide. \name{AgentSociety} accomplishes this by proving incentive compatibility in delegation to more competent neighbors based on local context obtained through incentivized peer communication, along with a payoff design that rewards the highest consensus based winning path with agents within receiving payoffs that are fair and representative of their marginal contribution.

\section{Experimental Results}
\label{results}

We conduct a comprehensive empirical evaluation of \name{AgentSociety} to demonstrate its efficacy in facilitating decentralized, consensus-based routing for user requests at test-time. Our analysis centers on the emergence of incentivized global information flow driven by the utility-maximizing strategic behavior of individual agents, where we characterize the mechanism, benchmark LLM agents on social intelligence and demonstrate that our mechanism enables self-interested autonomous agents to provide improved performance on real-world datasets through consensus-based routing.

\textbf{Mechanism Characterization.} We characterize \name{AgentSociety} with heterogeneous agents (up to 30) under best-response dynamics (please see \S~\ref{intra_br} for proof) across diverse graph topologies. Fig. \ref{fig:strategic_analysis} (left) illustrates that local utility maximization facilitates competence discovery, wherein self-interested delegation drives the network toward the global performance frontier. Mechanistic analysis (for node 10 in \S~\ref{single_agent_config}) reveals that payoffs are non-monotonic with respect to information diffusion: under-diffusion fails to establish leverage and ability to receive a payoff, while excessive diffusion erodes the agent’s marginal value as showin in Fig. \ref{fig:strategic_analysis} (right). These results demonstrate that \name{AgentSociety} successfully aligns individual incentives with collective efficiency, transforming a decentralized set of self-interested actors into a robust discovery engine that optimizes system-wide utility. We provide further experimental setup details in \S~\ref{mech}. 

\begin{figure}[!ht]
    \centering
    \small
    \begin{minipage}{0.45\textwidth}
        \centering
        \includegraphics[width=\linewidth]{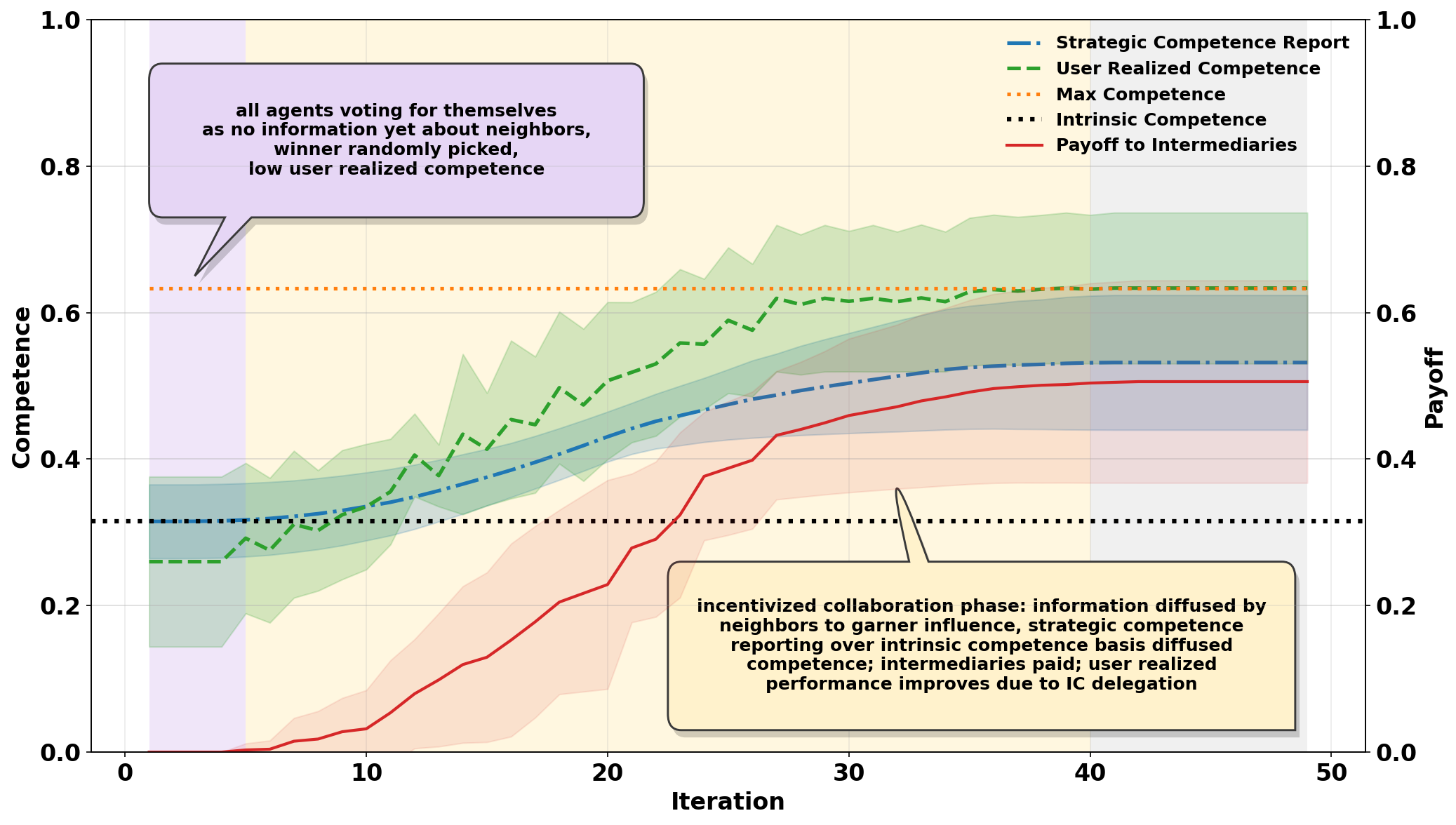}
    \end{minipage}
    \hfill
\begin{minipage}{0.5\textwidth}
    \centering
    \scriptsize 
    \setlength{\tabcolsep}{4pt}
    \renewcommand{\arraystretch}{1.2}
    \begin{tabular}{c c c c c}
        \toprule
        \makecell[b]{\textbf{Diffused} \\ \textbf{Competence}} & 
        \makecell[b]{\textbf{Votes} \\ \textbf{Garnered}} & 
        \makecell[b]{\textbf{In Feasible} \\ \textbf{Path}} & 
        \makecell[b]{\textbf{In Winning} \\ \textbf{Path}} & 
        \makecell[b]{\textbf{Payoff} \\ \textbf{Received}} \\
        \midrule
        <0.5 & 0 & $\times$ & $\times$ & 0 \\
        0.5    & 1 & $\times$ & $\times$ & 0 \\
        0.6    & 2 & \checkmark & $\times$ & 0 \\
        \textbf{0.7}    & 4 & \textbf{\checkmark} & \textbf{\checkmark} & \textbf{0.2} \\
        0.8    & 9 & \checkmark & \checkmark & 0.1 \\
        \bottomrule
    \end{tabular}
\end{minipage}
    
    \vspace{2mm} 
    \caption{\textbf{Characterization of \name{AgentSociety}.} (Left) Averaged over multiple graph configurations and competences, we demonstrate the mechanism driving user realized competence (green) higher through incentivized diffusion by intermediaries, reflected by their increasing payments (red). (Right) Impact of information diffusion intensity on routing outcomes for a node to demonstrate that strategic disclosure is necessary to transition from infeasibility to path selection and reward realization.}
    \label{fig:strategic_analysis}
\end{figure}


\textbf{LLM Agent Social Intelligence.} We utilize \name{AgentSociety} to empirically characterize the strategic behavior of LLM-based agents in terms of their capacity for selective information disclosure and calibration of competence reports in response to neighbor signals. This framing positions \name{AgentSociety} as a dual-purpose analytical tool: a diagnostic benchmark for individual LLM alignment with strategic optima, and a principled framework for evaluating the collaboration and aggregate performance of heterogeneous agents on real-world tasks discussed below. On the former, LLM agents are presented with the system prompt elaborating on the working of \name{AgentSociety} and a user prompt to obtain the LLM agent's response to delegation and diffusion (please see \S~\ref{app:prompts} for prompts). The LLM agent observes payoffs given by Eq. \ref{eq:del_payoff} with $\alpha$ set to 100. We quantify LLM behavioral divergence from the best-response equilibrium across two granularities in the presence of other best response agents: decision-level divergence, representing the cumulative difference in actions given a fixed state, and trajectory-level divergence, representing the cumulative shift in the evolved state itself. These deviations are measured over multiple configurations through the percentage overlap in delegation choices and the mean absolute error (MAE) of diffused information and reported competence in Fig. \ref{fig:combined_social_analysis} (left). We observe that Llama2.5-7bI, Qwen2.5-7bI, and GPT-4o-mini fail to obtain payoffs; GPT-5-mini, Gemini2.5-Flash, and Gemini2.5-Pro exhibit collaborative diffusion but over-disclose in enhanced configurations. Llama3.1-7b shows irrational self-delegation, while Qwen2.5-7b and GPT-4o-mini delegate irrationally occasionally; all other models maintain rationality. We provide a granular view for configuration \ref{conf1} in Fig. 
\ref{fig:combined_social_analysis} (right) with results for more configurations in \S~\ref{single_agent_config}. We provide further experimental setup details in \S~\ref{si}. 

\begin{figure}[!ht]
    \centering
    \small
    \begin{minipage}{0.58\textwidth}
        \centering
        \setlength{\tabcolsep}{4pt}
        \renewcommand{\arraystretch}{1.1}
        \begin{tabular}{l c c c}
            \toprule
            \textbf{Model} & \textbf{\makecell{Strat. Rep. \\ (MAE) $\downarrow$}} & \textbf{\makecell{Collab. Diff. \\ (MAE) $\downarrow$}} & \textbf{\makecell{I.C. Del. \\ Match \% $\uparrow$}} \\
            \midrule
            gpt-4o-mini      & 0.0668 & 0.0712 & 45.67\% \\
            llama-3.1-8b     & 0.0689 & 0.0674 & 46.67\% \\
            qwen-2.5-7b      & 0.0427 & 0.0574 & 47.00\% \\
            gemini-2.5-flash & 0.0368 & 0.0402 & 65.33\% \\
            gemini-2.5-pro   & \textbf{0.0304} & \textbf{0.0373} & 72.33\% \\
            gpt-5-mini       & 0.0325 & 0.0406 & \textbf{81.00\%} \\
            \midrule
            best-response    & 0.0000 & 0.0000 & 100.0\% \\
            \bottomrule
        \end{tabular}
    \end{minipage}
    \hfill
    \begin{minipage}{0.38\textwidth}
        \centering
        \includegraphics[width=\linewidth]{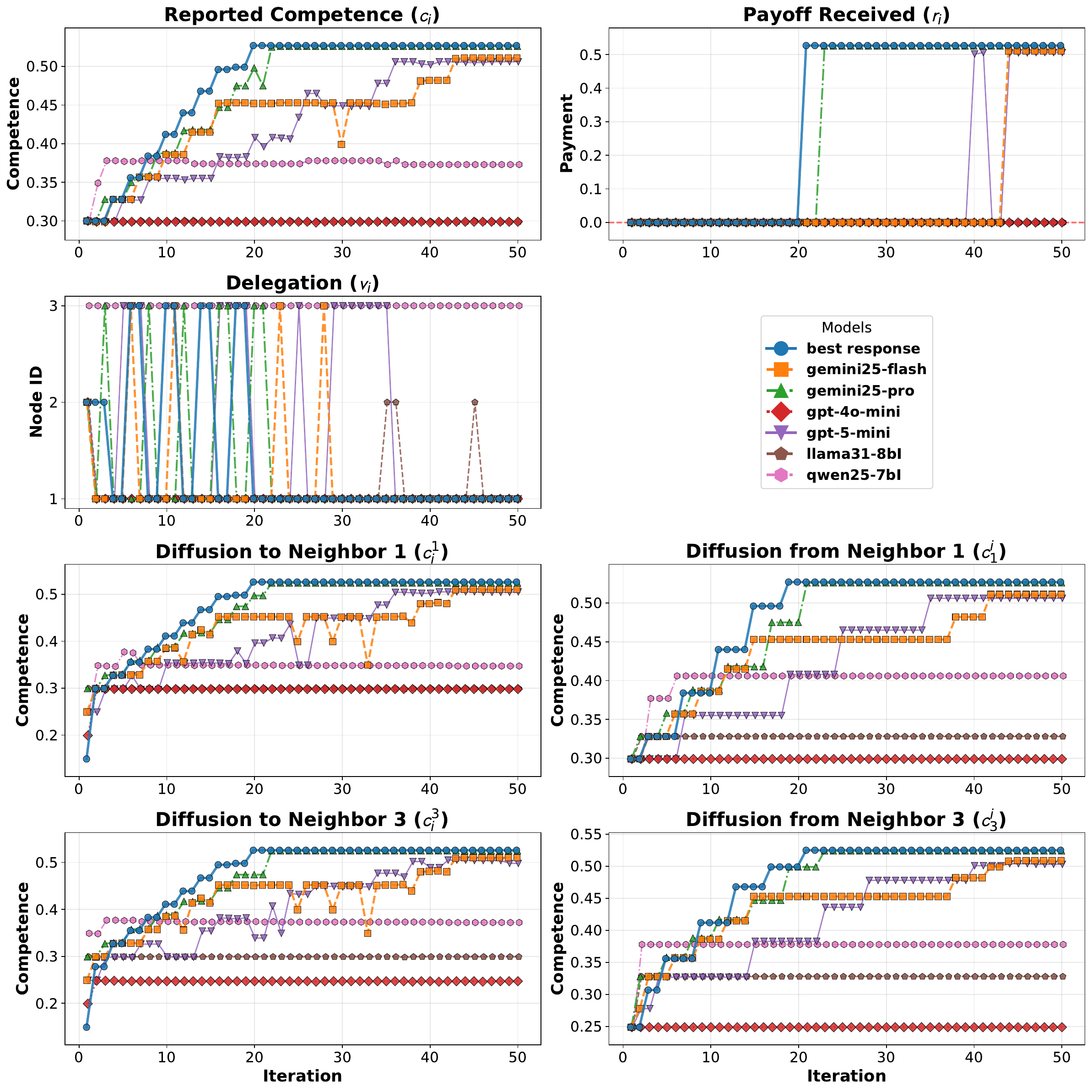}
    \end{minipage}
    \caption{\textbf{LLM Social Intelligence and Agent Dynamics.} Left: Quantitative deviation metrics for strategic reporting, collaborative diffusion, and incentive compatible delegation. Right: LLM agent dynamics against best-response for configuration 1, that provides a fine grained view into each agent's evolving actions (reported competence, delegation, diffusion to neighbor) and observations (payoff received, diffusion from)}
    \label{fig:combined_social_analysis}
\end{figure}

\textbf{Heterogeneous Collaborative MAS on Real-World Datasets.} We evaluate the collaborative performance of self-interested agents instantiated within \name{AgentSociety} on three benchmark suites: Open Leaderboard v2 (15 domains) \cite{openlb2}, MMLU-Pro (14 domains) \cite{mmlupro2}, and SWE-bench (5 domains) \cite{swe2}. Agent competence is parameterized as an $n$-dimensional vector, estimated on a 30\% training split and utilized during evaluation on the remaining 70\%. This parameterization serves as a generic proxy for competence making \name{AgentSociety} agnostic to the method of competence estimation. Table  \ref{tab:merged_results}(a) shows that \name{AgentSociety} yields consistent gains over the Best Single (BS) baseline, which approximates the performance frontier in non-incentivized settings. These improvements are driven by the novel interplay of incentivized information diffusion and delegation, allowing self-interested agents with specialized agent strengths to collaborate without a centralized authority.

\name{AgentSociety} scales seamlessly, both in number of agents with similar and varied capabilities. Table \ref{tab:merged_results}(b) shows that performance gains over BS scale with task complexity; as the number of domains or steps per request increases, the performance delta grows from +3.3\% to +7.1\%. This confirms that the mechanism's decentralized routing becomes increasingly critical as task complexity outpaces the capabilities of individual generalist models. Intermediary agents, driven by self-interest and information diffusion, ensure the routing of tasks to the most competent specialized models. On SWE-bench  in Table \ref{tab:merged_results}(a), we further demonstrate that the mechanism excels in identifying and incentivizing collaboration among heterogeneous agents with complementary strengths, particularly when using complementary base models. \name{AgentSociety} shines in incentivizing collaboration among heterogeneous agents with complementary strengths identified through social action and in each agent's self-interest. Oracle performance is obtained when each task is routed to the best domain model. The performance gap relative to the oracle is primarily attributable to two aspects, first being noisy, incomplete competence estimates from the training set, leading to competence inversions where training performance does not generalize to the test set. As shown in Table \ref{tab:merged_results}(c), performance monotonically approaches the oracle as estimate noise decreases. The second arises due to absence of feasible paths to reach the best agent for each task, for instance when a competent agent for a task lies within a clique and does not factor in the feasible paths $\mathbb{F}$. We provide further experimental details in \S~\ref{cps} and \S~\ref{cpm}.

\begin{table}[!ht]
\centering
\caption{\textbf{Real world dataset performance.} Part (a) shows \name{AgentSociety} benchmarks. Part (b) highlights scalability to multi-task and high-node (30 nodes) configurations. Part (c) illustrates performance under noisy estimates via a training ratio sweep.}
\label{tab:merged_results}
\small
\setlength{\tabcolsep}{4pt}
\renewcommand{\arraystretch}{1.05}
\begin{tabular}{@{}l c c c c c c@{}}
\toprule
            \textbf{Dataset/ Configuration} & \textbf{Agents} &  \textbf{Domains} &  \textbf{AS} & \textbf{BS} & \textbf{$\Delta$ (pp)} & \textbf{Oracle}\\ \midrule
                        MMLU-Pro & 6 & 14 & 0.7754 & 0.7568 & \textbf{+1.86} & 0.7821 \\

            Open LeaderBoard v2 & 6 & 15 & 0.6547 & 0.5714 & \textbf{+8.33} & 0.6565\\
            SWE-bench (Strong/ Generalist) & 6 & 5& 0.7341 & 0.7024 & \textbf{+3.2} & 0.7381 \\
            SWE-bench (Weaker/ Complementary) & 6 & 5&  0.5635 & 0.5079 & \textbf{+5.6} & 0.5873 \\ \midrule
            \textit{(Multi-task requests)} & & & & & &\\
            Two-task request & 30 & 5 & 0.6901 & 0.657 & \textbf{+3.31} & 0.708 \\
            Three-task request & 30 & 5 & 0.6963 & 0.642 & \textbf{+5.43} & 0.714 \\
            Four-task request & 30 & 5 & 0.7003 & 0.629 & \textbf{+7.1} & 0.718 
            \\ \midrule
            \textit{Noisy competence estimates (Open LeaderBoard v2)} & & & & & & \\
            Train ratio: 10\% & 6 & 15 & 0.6172 & 0.5714 & \textbf{+4.58} & 0.6565 \\
            Train ratio: 20\% & 6 & 15 &0.6393 & 0.5714 & \textbf{+6.79} & 0.6565  \\
            Train ratio: 25\% & 6 & 15 &0.6510 & 0.5714 & \textbf{+7.96} & 0.6565  \\
            Train ratio: 30\% & 6 & 15 &0.6547 & 0.5714 & \textbf{+8.33} & 0.6565 \\ \bottomrule
\end{tabular}
\end{table}

\paragraph{Limitations} Our mechanism relies on LLM agents being able to estimate their competence on a task, which currently remains an active domain of research \cite{casey}. However, as  our mechanism operates on competence as a generalized vector, we believe it is independent of and positioned to leverage these advances. We were unable to leverage frontier models for our evaluation on real world datasets due to their near-identical ranking across current datasets. In order to showcase improved collaboration through collaboration, we had to leverage weaker models that demonstrated complementary strengths.


\section{Conclusion}
We introduce a \name{AgentSociety}, a mechanism design grounded in social choice theory for LaMAS that enables agents to make autonomous decisions to maximize their utility while achieving collective objectives. We showed consensus based routing for agentic settings arises from delegation actions of agents. We proved that delegation to a more compatible neighbor is incentive compatible. Additionally, we showed the effectiveness of our mechanism in evaluation strategic behavior of agents critical for economic autonomy and frictionless interaction with current economy. Finally, we benchmark state-of-the-art LLMs with respect to best response. \name{AgentSociety} paves the path towards agentic economic autonomy and verifiable rewards using mechanism design. 

\newpage

\bibliography{references}

@inproceedings{
gptswarm,
title={{GPTS}warm: Language Agents as Optimizable Graphs},
author={Mingchen Zhuge and Wenyi Wang and Louis Kirsch and Francesco Faccio and Dmitrii Khizbullin and J{\"u}rgen Schmidhuber},
booktitle={Forty-first International Conference on Machine Learning},
year={2024},
url={https://openreview.net/forum?id=uTC9AFXIhg}
}

@inproceedings{
gdesigner,
title={G-Designer: Architecting Multi-agent Communication Topologies via Graph Neural Networks},
author={Guibin Zhang and Yanwei Yue and Xiangguo Sun and Guancheng Wan and Miao Yu and Junfeng Fang and Kun Wang and Tianlong Chen and Dawei Cheng},
booktitle={ICLR 2025 Workshop on Foundation Models in the Wild},
year={2025},
url={https://openreview.net/forum?id=Jov79pGXc6}
}

@inproceedings{
supernet,
title={Multi-agent Architecture Search via Agentic Supernet},
author={Guibin Zhang and Luyang Niu and Junfeng Fang and Kun Wang and LEI BAI and Xiang Wang},
booktitle={Forty-second International Conference on Machine Learning},
year={2025},
url={https://openreview.net/forum?id=imcyVlzpXh}
}

@misc{agentnet,
      title={AgentNet: Decentralized Evolutionary Coordination for LLM-based Multi-Agent Systems}, 
      author={Yingxuan Yang and Huacan Chai and Shuai Shao and Yuanyi Song and Siyuan Qi and Renting Rui and Weinan Zhang},
      year={2025},
      eprint={2504.00587},
      archivePrefix={arXiv},
      primaryClass={cs.MA},
      url={https://arxiv.org/abs/2504.00587}, 
}

@article{social-action,
title = {Modelling social action for AI agents},
journal = {Artificial Intelligence},
volume = {103},
number = {1},
pages = {157-182},
year = {1998},
note = {Artificial Intelligence 40 years later},
issn = {0004-3702},
doi = {https://doi.org/10.1016/S0004-3702(98)00056-3},
url = {https://www.sciencedirect.com/science/article/pii/S0004370298000563},
author = {Cristiano Castelfranchi},
}

@inproceedings{
aflow,
title={{AF}low: Automating Agentic Workflow Generation},
author={Jiayi Zhang and Jinyu Xiang and Zhaoyang Yu and Fengwei Teng and Xiong-Hui Chen and Jiaqi Chen and Mingchen Zhuge and Xin Cheng and Sirui Hong and Jinlin Wang and Bingnan Zheng and Bang Liu and Yuyu Luo and Chenglin Wu},
booktitle={The Thirteenth International Conference on Learning Representations},
year={2025},
url={https://openreview.net/forum?id=z5uVAKwmjf}
}

@misc{vae,
      title={Virtual Agent Economies}, 
      author={Nenad Tomasev and Matija Franklin and Joel Z. Leibo and Julian Jacobs and William A. Cunningham and Iason Gabriel and Simon Osindero},
      year={2025},
      eprint={2509.10147},
      archivePrefix={arXiv},
      primaryClass={cs.AI},
      url={https://arxiv.org/abs/2509.10147}, 
}

@misc{aex,
      title={Agent Exchange: Shaping the Future of AI Agent Economics}, 
      author={Yingxuan Yang and Ying Wen and Jun Wang and Weinan Zhang},
      year={2025},
      eprint={2507.03904},
      archivePrefix={arXiv},
      primaryClass={cs.AI},
      url={https://arxiv.org/abs/2507.03904}, 
}

@misc{mas-gpt,
      title={MAS-GPT: Training LLMs to Build LLM-based Multi-Agent Systems}, 
      author={Rui Ye and Shuo Tang and Rui Ge and Yaxin Du and Zhenfei Yin and Siheng Chen and Jing Shao},
      year={2025},
      eprint={2503.03686},
      archivePrefix={arXiv},
      primaryClass={cs.CL},
      url={https://arxiv.org/abs/2503.03686}, 
}

@misc{fermat,
      title={Discovering Coordinated Joint Options via Inter-Agent Relative Dynamics}, 
      author={Raul D. Steleac and Mohan Sridharan and David Abel},
      year={2025},
      eprint={2512.24827},
      archivePrefix={arXiv},
      primaryClass={cs.LG},
      url={https://arxiv.org/abs/2512.24827}, 
}

@inproceedings{
decpomdp1,
title={From Debate to Equilibrium: Belief\nobreakdash-Driven Multi\nobreakdash-Agent {LLM} Reasoning via Bayesian Nash Equilibrium},
author={Xie Yi and Zhanke Zhou and Chentao Cao and Qiyu Niu and Tongliang Liu and Bo Han},
booktitle={Forty-second International Conference on Machine Learning},
year={2025},
url={https://openreview.net/forum?id=RQwexjUCxm}
}

@misc{decpomdp2,
      title={LLM Collaboration With Multi-Agent Reinforcement Learning}, 
      author={Shuo Liu and Tianle Chen and Zeyu Liang and Xueguang Lyu and Christopher Amato},
      year={2025},
      eprint={2508.04652},
      archivePrefix={arXiv},
      primaryClass={cs.AI},
      url={https://arxiv.org/abs/2508.04652}, 
}

@inproceedings{
cutthecrap,
title={Cut the Crap: An Economical Communication Pipeline for {LLM}-based Multi-Agent Systems},
author={Guibin Zhang and Yanwei Yue and Zhixun Li and Sukwon Yun and Guancheng Wan and Kun Wang and Dawei Cheng and Jeffrey Xu Yu and Tianlong Chen},
booktitle={The Thirteenth International Conference on Learning Representations},
year={2025},
url={https://openreview.net/forum?id=LkzuPorQ5L}
}

@inproceedings{viscous2,
  title     = {Optimizing Viscous Democracy},
  author    = {Armstrong, Ben and Alouf-Heffetz, Shiri and Talmon, Nimrod},
  booktitle = {Proceedings of the Thirty-Third International Joint Conference on
               Artificial Intelligence, {IJCAI-24}},
  publisher = {International Joint Conferences on Artificial Intelligence Organization},
  editor    = {Kate Larson},
  pages     = {2643--2650},
  year      = {2024},
  month     = {8},
  note      = {Main Track},
  doi       = {10.24963/ijcai.2024/292},
  url       = {https://doi.org/10.24963/ijcai.2024/292},
}

@article{liqdem2,
author = {G\"{O}lz, Paul and Kahng, Anson and Mackenzie, Simon and Procaccia, Ariel D.},
title = {The Fluid Mechanics of Liquid Democracy},
year = {2021},
issue_date = {December 2021},
publisher = {Association for Computing Machinery},
address = {New York, NY, USA},
volume = {9},
number = {4},
issn = {2167-8375},
url = {https://doi.org/10.1145/3485012},
doi = {10.1145/3485012},
journal = {ACM Trans. Econ. Comput.},
month = oct,
articleno = {23},
numpages = {39}
}

@article{boldi,
author = {Boldi, Paolo and Bonchi, Francesco and Castillo, Carlos and Vigna, Sebastiano},
title = {Viscous democracy for social networks},
year = {2011},
issue_date = {June 2011},
publisher = {Association for Computing Machinery},
address = {New York, NY, USA},
volume = {54},
number = {6},
issn = {0001-0782},
url = {https://doi.org/10.1145/1953122.1953154},
doi = {10.1145/1953122.1953154},
journal = {Commun. ACM},
month = jun,
pages = {129–137},
numpages = {9}
}

@article{liqdem1, title={Liquid Democracy: An Algorithmic Perspective}, volume={32}, url={https://ojs.aaai.org/index.php/AAAI/article/view/11468}, DOI={10.1609/aaai.v32i1.11468}, number={1}, journal={Proceedings of the AAAI Conference on Artificial Intelligence}, author={Kahng, Anson and Mackenzie, Simon and Procaccia, Ariel}, year={2018}, month={Apr.} }

@inproceedings{mdsn,
  title={Mechanism design in social networks},
  author={Li, Bin and Hao, Dong and Zhao, Dengji and Zhou, Tao},
  booktitle={Proceedings of the AAAI Conference on Artificial Intelligence},
  volume={31},
  number={1},
  year={2017}
}

@misc{idm,
      title={Mechanism Design in Social Networks}, 
      author={Bin Li and Dong Hao and Dengji Zhao and Tao Zhou},
      year={2017},
      eprint={1702.03627},
      archivePrefix={arXiv},
      primaryClass={cs.GT},
      url={https://arxiv.org/abs/1702.03627}, 
}

@inproceedings{
pokemon,
title={Pok\'eChamp: an Expert-level Minimax Language Agent},
author={Seth Karten and Andy Luu Nguyen and Chi Jin},
booktitle={Forty-second International Conference on Machine Learning},
year={2025},
url={https://openreview.net/forum?id=SnZ7SKykHh}
}

@inproceedings{
board-games,
title={Mastering Board Games by External and Internal Planning with Language Models},
author={John Schultz and Jakub Adamek and Matej Jusup and Marc Lanctot and Michael Kaisers and Sarah Perrin and Daniel Hennes and Jeremy Shar and Cannada A. Lewis and Anian Ruoss and Tom Zahavy and Petar Veli{\v{c}}kovi{\'c} and Laurel Prince and Satinder Singh and Eric Malmi and Nenad Tomasev},
booktitle={Forty-second International Conference on Machine Learning},
year={2025},
url={https://openreview.net/forum?id=KKwBo3u3IW}
}

@inproceedings{
duan2024gtbench,
title={{GTB}ench: Uncovering the Strategic Reasoning Capabilities of {LLM}s via Game-Theoretic Evaluations},
author={Jinhao Duan and Renming Zhang and James Diffenderfer and Bhavya Kailkhura and Lichao Sun and Elias Stengel-Eskin and Mohit Bansal and Tianlong Chen and Kaidi Xu},
booktitle={The Thirty-eighth Annual Conference on Neural Information Processing Systems},
year={2024},
url={https://openreview.net/forum?id=ypggxVWIv2}
}

@inproceedings{huang2025competing,
  title={Competing large language models in multi-agent gaming environments},
  author={Huang, Jen-tse and Li, Eric John and Lam, Man Ho and Liang, Tian and Wang, Wenxuan and Yuan, Youliang and Jiao, Wenxiang and Wang, Xing and Tu, Zhaopeng and Lyu, Michael},
  booktitle={The Thirteenth International Conference on Learning Representations},
  year={2025}
}

@inproceedings{
govsim,
title={Cooperate or Collapse:  Emergence of Sustainable Cooperation in a Society of {LLM} Agents},
author={Giorgio Piatti and Zhijing Jin and Max Kleiman-Weiner and Bernhard Sch{\"o}lkopf and Mrinmaya Sachan and Rada Mihalcea},
booktitle={The Thirty-eighth Annual Conference on Neural Information Processing Systems},
year={2024},
url={https://openreview.net/forum?id=0zWzJj6lO3}
}

@inproceedings{
casey,
title={Do Large Language Models Know What They Are Capable Of?},
author={Casey O. Barkan and Sidney Black and Oliver Sourbut},
booktitle={The Fourteenth International Conference on Learning Representations},
year={2026},
url={https://openreview.net/forum?id=EO6WtJ0q6G}
}

@misc{mmlupro2,
      title={MMLU-Pro: A More Robust and Challenging Multi-Task Language Understanding Benchmark}, 
      author={Yubo Wang and Xueguang Ma and Ge Zhang and Yuansheng Ni and Abhranil Chandra and Shiguang Guo and Weiming Ren and Aaran Arulraj and Xuan He and Ziyan Jiang and Tianle Li and Max Ku and Kai Wang and Alex Zhuang and Rongqi Fan and Xiang Yue and Wenhu Chen},
      year={2024},
      eprint={2406.01574},
      archivePrefix={arXiv},
      primaryClass={cs.CL},
      url={https://arxiv.org/abs/2406.01574}, 
}

@misc{swe2,
      title={SWE-bench: Can Language Models Resolve Real-World GitHub Issues?}, 
      author={Carlos E. Jimenez and John Yang and Alexander Wettig and Shunyu Yao and Kexin Pei and Ofir Press and Karthik Narasimhan},
      year={2024},
      eprint={2310.06770},
      archivePrefix={arXiv},
      primaryClass={cs.CL},
      url={https://arxiv.org/abs/2310.06770}, 
}

@misc{openlb2,
  author = {Clémentine Fourrier and Nathan Habib and Alina Lozovskaya and Konrad Szafer and Thomas Wolf},
  title = {Open LLM Leaderboard v2},
  year = {2024},
  publisher = {Hugging Face},
  howpublished = {\url{https://huggingface.co/spaces/open-llm-leaderboard/open_llm_leaderboard}}
}

@misc{mass,
      title={Social Theory Should Be a Structural Prior for Agentic AI: A Formal Framework for Multi-Agent Social Systems}, 
      author={Lynnette Hui Xian Ng and Iain J. Cruickshank and Adrian Xuan Wei Lim and Kathleen M. Carley},
      year={2026},
      eprint={2605.07069},
      archivePrefix={arXiv},
      primaryClass={cs.MA},
      url={https://arxiv.org/abs/2605.07069}, 
}


\newpage

\appendix

\section{Proofs of Main Results}
We present the proofs of our theorems and lemmas here.

\subsection{Proof for Incentive Compatibility in Delegation}
\label{ic}
\begin{proof}[Proof of Theorem~\ref{thm:ic}]
All agents act rationally. Hence, at each time step \( \tau \), the payment received by any
agent \( n_i \in \mathcal{N} \) is determined entirely by the first component of the payment
rule, since any deviation that triggers an infinite penalty is strictly dominated.
Consequently, agents have no incentive to misreport in ways that violate feasibility.

Fix a task \( t_i \in \mathcal{T} \), and consider an agent \( n_i \) that lies on the critical
path for this task.
Conditioned on the reported competence vector
\(
(c'_{i,1},\dots,c'_{i,k}),
\)
agent \( n_i \) remains on the winning path regardless of its delegation decision
\( v_i^\tau \in r_i^\tau \cup \{n_i\} \). For agents not on the critical path, the payment remains zero irrespective of their
delegation choice. The realized payment to agent \( n_i \) depends only on the set of agents assigned
to task \( t_i \) and their reported competences.

Formally, for any two delegation choices
\( v_i^\tau\) of the remaining agents, we have
\[
\mathbb{I}\!\left\{ n_i \in W(v_i^\tau, a_{-i}^\tau) \right\}
=
\mathbb{I}\!\left\{ n_i \in W(\tilde v_i^\tau, a_{-i}^\tau) \right\},
\]
where \( W(\cdot) \) denotes the winning path induced by the corresponding joint action.
Thus, delegation preserves the voting weight accumulated by agent \( n_i \), since votes from
upstream agents remain unchanged and agent \( n_i \) retains its competence contribution on
all tasks.

We now analyze the incentive properties of agent \( n_i \) under different reporting behaviors.

\paragraph{Case 1: Truthful reporting.}
Suppose agent \( n_i \) reports its true competence on task \( t_i \), i.e.,
\( c'_{i,t_i} = c_{i,t_i} \).
Then, for any feasible delegation decision
\( v_i^\tau\),
the utility of agent \( i \) satisfies
$
u_i(a_i^\tau, a_{-i}^\tau)
=
f(c_{i,t_i}) - f(c_{i-1,t_i}),
$
which is independent of the choice of \( v_i^\tau \).
Hence, truthful reporting yields identical utility across all feasible delegation actions.

\paragraph{Case 2: Over-reporting.}
Suppose agent \( i \) inflates its reported competence from its true value \( c_{i,t_i} \) to
some \( c'_{i,t_i} > c_{i,t_i} \), based on its private social information
\( h_i^\tau \).
Such over-reporting is feasible only if there exists a neighbor
\( n_j \in r_i^\tau \) such that
$
c_{j \to i, t_k}^{\tau-1}
>
c_{i,t_i}.
$ If agent \( n_i \) chooses self-delegation, \( v_i^\tau = n_i \), then no such informational
advantage can be exploited, and thus \( c'_{i,t_i} = c_{i,t_i} \).
In this case, the utility of agent \( n_i \) is
$
u_i^{\mathrm{self}}
=
f(c_{i,t_i}) - f(c_{i-1,t_i}).
$ Now suppose agent \( n_i \) delegates to a neighbor
\( j \in r_i^\tau \) such that
$
c_{j \to i, t_k}^{\tau-1}
>
c_{i,t_i}.
$ Then agent \( n_i \) may inflate its reported competence up to
$
c'_{i,t_i}
\le
c_{j \to i,,t_i}^{\,\tau-1}.
$
In this case, the highest reported competence on task \( t_i \) becomes
\(c_{j \to i,,t_i}^{\,\tau-1} \), yielding utility
$
u_i^{\mathrm{del}}
=
f(c'_{i,t_i}) - f(c_{i-1,t_i}).
$

If
$
\max_{j \in r_i^\tau \cup \{i\}}
c_{j \to i, t_k}^{\tau-1}
=
c_{i,t_i},
$
then no such neighbor exists and \( c'_{i,t_i} = c_{i,t_i} \), implying
\(
u_i^{\mathrm{self}} = u_i^{\mathrm{del}}.
\)
Otherwise, if there exists a neighbor \( n_j \) satisfying
\(
c_{j \to i, t_k}^{\tau-1},
\)
then delegating to \( n_j \) allows agent \( n_i \) to increase its reported competence, yielding
$
u_i^{\mathrm{del}}
\;\ge\;
u_i^{\mathrm{self}}.
$
Therefore, delegation to a neighbor with higher diffused competence weakly dominates
self-voting under over-reporting.

\paragraph{Case 3: Under-reporting.}
Suppose agent \( n_i \) reports \( c'_{i,t_i} < c_{i,t_i} \).
Then its utility becomes
$
u_i^{\mathrm{self}}
=
f(c'_{i,t_i}) - f(c_{i-1,t_i}),
$
which is strictly smaller than the truthful utility
$
f(c_{i,t_i}) - f(c_{i-1,t_i}),
$
since \( f(\cdot) \) is strictly increasing.
Therefore, under-reporting strictly reduces utility and is never optimal for a rational
agent.

Combining all three cases, we conclude that delegating to a neighbor \( n_j \) such that
\(
c_{j \to i, t_k}^{\tau-1} > c_{i,t_i},
\) constitutes a weakly dominant strategy relative to retaining
the vote. This completes the proof.
\end{proof}

\newpage
\subsection{Proof for Continuous Critical Path}
\label{proof-critical}
\begin{proof}[Proof of Theorem~\ref{thm:critical}]
Fix a task $t_k \in \mathcal{T}$. For each agent $n_i$, let $v_i$ denote its delegated vote, and let $\mathrm{guru}(i)$ be the terminal agent reached by following delegation links from $n_i$.

Let $C_{t_k}$ be the set of agents critical for $t_k$, where $n_i$ is critical if, when $n_i$ deviates by voting for itself, the resulting winning path contains (and thus terminates at) $n_i$.

First, let $n_a,n_b \in C_{t_k}$. Suppose for contradiction that
\[
\mathrm{guru}(a) \neq \mathrm{guru}(b).
\]
If $n_a$ votes for itself, criticality implies the winning path terminates at $n_a$. Similarly, if $n_b$ votes for itself, it must terminate at $n_b$. Since all other agents’ actions remain unchanged in each deviation, both outcomes cannot simultaneously correspond to strictly maximal total weight. This contradicts the assumption that both $n_a$ and $n_b$ are critical. Hence,
\[
\mathrm{guru}(a) = \mathrm{guru}(b).
\]

It remains to prove that $C_{t_k}$ is contiguous along the winning delegation path. Suppose $n_a,n_b \in C_{t_k}$ share the same guru, but neither lies on the delegation path of the other. Consider the deviation where $n_a$ votes for itself. By criticality, the winning path must terminate at $n_a$, while $n_b$’s contribution remains unchanged. The symmetric argument holds when $n_b$ votes for itself.

Thus, the two deviations produce winning paths ending at different agents with identical contributions from all non-deviators, so both cannot yield strictly maximal weight. This contradiction implies that if two agents in $C_{t_k}$ share the same guru, one must lie on the delegation path of the other. Therefore, $C_{t_k}$ forms a continuous segment of the winning delegation path.
\end{proof}

\subsection{Proof of Nash Equilibrium Theorem \ref{thm:ne}}
\label{ne}

\subsubsection{Nash Equilibrium Delegation Rule}

Add what the equilibrium rule is here

Fix a Nash equilibrium action profile $(a_i)_{n_i \in \mathcal{N}}$ and let $n_i \in \mathcal{N}$
be an arbitrary agent. After the joint action $a = (a_1, \ldots, a_N)$ is taken and the state
transitions to $s'$, agent $n_i$ receives a private observation
\[
    o_i \;\in\; \mathcal{O}_i,
\]
comprising its realized payment together with private social information
$\{c_{j \to i,\,t_k}^{\tau-1}\}$, formally:
\[
    o_i \;=\; \Bigl\{\,\bigl(j,\;(c_{j \to i,1},\ldots,c_{j \to i,k}),\;h_j^i\bigr)
    \;\Big|\; j \in r_i \cup \{i\} \Bigr\}.
\]
In particular, for each neighbor $n_j \in r_i$, agent $n_i$ observes the diffused competence
$c_{j \to i,\,t_i}$ on its assigned task $t_i$, and observes its own true competence
$c_{t_i}^{i \to i} = c_i$. Define
\[
    \tilde{c}_i \;=\; \max_{j \in r_i}\; c_{j \to i,\,t_i}
\]
as the highest diffused competence value observed by $n_i$ across its neighborhood $r_i$.
We establish the equilibrium delegation rule by exhaustive case analysis.

\medskip
\noindent\textbf{Case 1: $c_i \geq \tilde{c}_i$.}

Suppose, toward contradiction, that at equilibrium agent $n_i$ delegates to some neighbor,
i.e.\ $v_i = n_j$ for some $j \in r_i$. By definition of $\tilde{c}_i$, the diffused
competence of any such neighbor satisfies
\[
    c_{j \to i,\,t_i} \;\leq\; \tilde{c}_i \;\leq\; c_i.
\]
Hence the guru induced by the delegation path $v_i = n_j$ has reported competence on task
$t_i$ that is weakly less than $c_i$. By the payment rule, agent $n_i$'s reported competence
weakly exceeds that of its guru, so its marginal contribution is non-positive and its payment
is weakly negative.

In contrast, if agent $n_i$ self-delegates ($v_i = n_i$), the reporting rule yields
\[
    c_{t_i}^{i \to i} \;=\; c_i,
\]
so $n_i$ incurs no negative payment. Since self-delegation weakly dominates delegating to any
neighbor, the unique best response is
\[
    v_i \;=\; n_i.
\]

\medskip
\noindent\textbf{Case 2: $c_i < \tilde{c}_i$.}

Let
\[
    n_{j^*} \;\in\; \arg\max_{j \in r_i}\; c_{j \to i,\,t_i}
\]
denote a neighbor achieving the maximal diffused competence observed by $n_i$. We rule out
the two suboptimal alternatives.

\smallskip
\noindent\textit{Subcase 2a: Delegating to a non-maximizing neighbor.}
If $v_i = n_j$ for some $j \neq j^*$, the utility received by $n_i$ is
\[
    f\!\left(c_{j \to i,\,t_i}\right) - f\!\left(c'_{i-1,\,t_i}\right)
    \;<\;
    f\!\left(c_{j^* \to i,\,t_i}\right) - f\!\left(c'_{i-1,\,t_i}\right),
\]
since $c_{j \to i,\,t_i} < c_{j^* \to i,\,t_i}$ and $f$ is strictly increasing. Hence
delegating to any $n_j$ with $j \neq j^*$ is strictly dominated by delegating to $n_{j^*}$.

\smallskip
\noindent\textit{Subcase 2b: Self-delegation.}
If $v_i = n_i$, the reporting rule yields $c_{t_i}^{i \to i} = c_i$. Since
$c_i < \tilde{c}_i = c_{j^* \to i,\,t_i}$, self-delegation produces strictly lower utility
than delegating to $n_{j^*}$.

\smallskip
Combining both subcases, the unique best response is
\[
    v_i \;=\; j^* \;\in\; \arg\max_{j \in r_i}\; c_{j \to i,\,t_i}.
\]

\medskip
\noindent\textbf{Conclusion.}

Since Cases~1 and~2 are exhaustive and mutually exclusive, at any Nash equilibrium action
profile $(a_i)_{n_i \in \mathcal{N}}$ every agent $n_i \in \mathcal{N}$ adopts the delegation
rule
\[
    v_i \;=\;
    \begin{cases}
        n_i,
            & \text{if } c_i \geq \tilde{c}_i, \\[6pt]
        \displaystyle\arg\max_{n_j \in r_i}\; c_{j \to i,\,t_i},
            & \text{if } c_i < \tilde{c}_i.
    \end{cases}
\]
\hfill$\blacksquare$

\paragraph{Proof.}
Fix a Nash equilibrium action profile \( (a_{i})_{n_i \in \mathcal{N}} \), and fix agent
\( n_i \in \mathcal{N} \).

\medskip
\noindent\textbf{\textit{Case 1: Diffusion to the critical agents}}

If agent \( n_i \)'s private social action \( h_i \) fails to diffuse a competence value high
enough to the critical agent, then that agent does not choose \( v_j = n_i \).
Consequently, agent \( n_i \) does not lie on the winning path and receives zero payment:
\[
u_i(a_i, a_{-i}) = 0.
\]

Thus, any equilibrium diffusion must transmit sufficient competence information to the
critical agent to secure delegation. The minimum such diffusion is $c_{t_i}^{(2)}$. Applying this inductively across nodes we get the equilibrium reporting strategy as \[
c_{t_i}^{\,h} =
\begin{cases}
c_{t_i}^{\,i}, & v_i = n_i \\[6pt]
\hat{c_i}, & v_i \neq n_i
\end{cases}
\]
where $\hat c_i
=
\max_{v \in S(i,v_i)\setminus
\arg\max_{u \in S(i,v_i)} c_u} c_v$, 
$\tilde d_i = \max_{j \in r_i} c_{t_i}^{j \to i}$ and
$S(i,j) = \{ n_k \in \mathcal{N} \;|\; n_j \in \operatorname{reach}(n_k) \;\text{and}\; n_k \in P,\; \forall P \in \mathcal{P}_{j \to i} \}$,
where $\operatorname{reach}(n_i) = \{ n_j \in \mathcal{N} \mid \exists \text{ a directed path from } n_i \text{ to } n_j \}$ denotes the set of agents reachable from $n_i$, and $\mathcal{P}_{j \to i}$ denotes the set of all directed paths from $n_j$ to $n_i$ in the diffusion graph. Intuitively, this means that the most competent reachable agent $n_j$ to agent $n_i$ diffuses just enough information, which is equivalent to most competent reachable node not in the path of $n_i, n_j$ to get the vote of $n_i$.

\medskip
\noindent\textbf{\textit{Case 2: Diffusion to non-critical agents}}

Altering diffusion to agents not on the critical path does not change:
the identity of the guru, the reported values \( c_{t_i}^{i \to i} \) and
\( c_{t_i}^{(2)} \), or the winning path. Hence,
\[
u_i(a_i, a_{-i}) = f(c_{t_i}^{i \to i}) - f(c_{t_i}^{(2)})
\]
remains unchanged.

\medskip
\noindent\textbf{\textit{Step 3: Optimality of the diffusion rule.}}

Any best response diffusion must:
\begin{itemize}
  \item ensure delegation by the critical agent;
  \item not reduce the agent’s own payment.
\end{itemize}

The diffusion rule
\[
c_{i \to j,t_j} =
\begin{cases}
c_{i \to i,t_i}, & \text{if } c_{j \to i,t_j} > c_{i \to i,t_j}, \\[4pt]
c_{j \to i,t_j}, & \text{if } c_{j \to i,t_j} \le c_{i \to i,t_j}.
\end{cases}
\]
satisfies both conditions. It guarantees delegation from all agents with higher competence
while preserving the payment term \( f(c_{t_i}^{i \to i}) - f(c_{t_i}^{(2)}) \). Hence, this diffusion strategy is optimal at Nash equilibrium

As in the equilibrium described if agent \( n_i \) delegates its vote to \( g_i \), then every agent reachable from \( n_i \)
whose task coincides with that of \( n_i \) also delegates to \( g_i \).
Consequently, all such agents lie on the same feasible path, if any at equilibrium.
This ensures that agent \( n_i \) can report the highest achievable competence on all auxiliary
tasks, since all reachable agents contributing to those tasks are contained within a single
path. Any deviation from this delegation structure strictly reduces the set of agents on agent
\( n_i \)’s feasible path.
As a result, the maximal reported competence attainable on auxiliary tasks weakly decreases.
Therefore, no such deviation can improve agent \( n_i \)’s outcome, and the preceding analysis
remains unaffected.

\hfill \( \square \)

\subsection{Best Response for Single-Task Request}
\label{intra_br}
\begin{lemma}[Best Response for Single-Task Request]
\label{lem:best_response_diffusion}
\textbf{Diffusion and Delegation.}
Let agent \( n_i \in \mathcal{N} \) have primary task \( t_i \), neighbors \( r_i \), and true
competence \( c_{i,t_i} \) on task \( t_i \).
At round \( \tau-1 \), agent \( n_i \) observes from its private observation, for each
\( n_j \in r_i \), the diffused competence value
\[
c_{j \to i,t_i}^{\,\tau-1},
\text{ with } c_{i \to i,t_i}^{\,\tau-1} := c_{i,t_i}.
\]

We define
\[
c^*
=
\max_{n_j \in r_i \cup \{n_i\}} c_{j \to i,t_i}^{\,\tau-1}, 
n_{j^*}
\in
\arg\max_{n_j \in r_i \cup \{n_i\}} c_{j \to i,t_i}^{\,\tau-1}.
\]

Then there exists a best response of agent \( n_i \) at round \( \tau \) such that:

\medskip
\noindent\textbf{(Diffusion)}
The private social action \( h_i^\tau \in \mathcal{H}_i \) satisfies, for each
\( n_j \in r_i \),
\[
c_{i \to j,,t_i}^{\,\tau}
=
\begin{cases}
\varnothing, & n_j = n_{j^*}, \\[4pt]
\min\!\left(c_{j \to i,t_i}^{\,\tau-1} + \delta,\; c^* \right),
& n_j \in r_i \setminus \{j^*\} \text{ and } c_{j \to i,t_i}^{\,\tau-1} < c^*, \\[6pt]
\varnothing, & \text{otherwise},
\end{cases}
\]
where \( \delta > 0 \) is a fixed diffusion increment, and
\( c_{i \to j,t_i}^{\,\tau} \) denotes the competence component on task \( t_i \)
transmitted by agent \( n_i \) to agent \( n_j \) through \( h_i^\tau \).

\medskip
\noindent\textbf{(Delegation)}
The delegation component of agent \( n_i \)'s action satisfies
\[
v_i^\tau \in \arg\max_{v \in r_i \cup \{i\}} c_{v \to i,,t_i}^{\,\tau-1}.
\]

In particular, delegating to an agent attaining the maximal observed diffused competence
value (or to itself, if it attains the maximum), together with the diffusion rule above,
constitutes a best response for agent \( n_i \) at round \( \tau \).
\end{lemma}

\paragraph{Delegation Best Response}

Fix any agent \( n_i \in \mathcal{N} \), and let \( r_i \) denote its neighbors.
From its private observation, agent \( n_i \) observes, for each \( n_j \in r_i \), the diffused
competence value \( c_{j \to i,t_i} \) on its own task \( t_i \), and it observes its own
true competence \( c_{i \to i,t_i} = c_i \). We define
$
\hat c_i
=
\max\!\Big( c_i,\; \max_{j \in r_i} c_{j \to i,t_i} \Big),
$
the highest competence value available to agent \( n_i \). Let agent \( n_i \) choose its delegation component
$
v_i^*
\in
\arg\max_{v \in r_i \cup \{i\}} c_{t_i}^{v \to i}.
$ We show that this delegation rule is optimal: for any alternative delegation
\( v_i' \neq v_i^* \), i.e.
$
u_i\!\big(a_i(v_i^*), a_{-i}\big)
\;\ge\;
u_i\!\big(a_i(v_i'), a_{-i}\big),
$
where \( a_i(v) \) denotes agent \( n_i \)'s action with delegation component \( v \)
and all other components fixed.

\medskip
Agent \( n_i \) receives a strictly positive payment if and only if it lies on the
critical path of the winning delegation tree.

\medskip
\noindent\textbf{\textit{Case 1: Agent \( n_i \) is not critical under \( v_i^* \).}}

Any alternative delegation \( v_i' \neq v_i^* \) cannot move agent \( n_i \) onto the winning
path, and hence
\[
u_i\!\big(a_i(v_i'), a_{-i}\big)
\le
0
=
u_i\!\big(a_i(v_i^*), a_{-i}\big).
\]

\medskip
\noindent\textbf{\textit{Case 2: Agent \( n_i \) is critical under \( v_i^* \).}}

Then agent \( n_i \) lies on the winning path, and its realized payment equals its marginal
contribution:
\[
u_i\!\big(a_i(v_i^*), a_{-i}\big)
=
f(\hat c_i) - f(\hat c_i^{(2)}),
\]
where \( \hat c_i^{(2)} \) denotes the competence value of the agent immediately preceding
\( n_i \) on the critical path, if any. Consider any alternative delegation \( v_i' \neq v_i^* \).
By definition of \( v_i^* \),
$
c_{v_i' \to i,t_i} \le c_{v_i^* \to i,t_i} = \hat c_i.
$ Since agent \( n_i \) remains critical under \( v_i' \), its payment is
\[
u_i\!\big(a_i(v_i'), a_{-i}\big)
=
f(c_{v_i' \to i,t_i}) - f(\hat c_i^{(2)}).
\]

Because \( f \) is monotone non-decreasing,
$
f(c_{v_i' \to i,t_i}) \le f(\hat c_i).
$
Hence,
\[
u_i\!\big(a_i(v_i'), a_{-i}\big)
\le
f(\hat c_i) - f(\hat c_i^{(2)})
=
u_i\!\big(a_i(v_i^*), a_{-i}\big).
\]

Therefore, any deviation from \( v_i^* \) weakly decreases agent \( n_i \)'s payment, with
strict inequality whenever
\(
c_{v_i' \to i,t_i} < c_{v_i^* \to i,t_i}.
\)

Thus, delegating to
$
v_i^* \in \arg\max_{v \in r_i \cup \{i\}} c_{v \to i,t_i}
$ is optimal for agent \( n_i \).
\hfill $\square$

\medskip
\paragraph{Diffusion Best Response}
From its private observation at round \( \tau-1 \), agent \( n_i \) observes
\[
\{\, c_{j \to i,t_i}^{\,\tau-1} \mid n_j \in r_i \,\},
\text{ with } c_{i \to i,t_i}^{\,\tau-1} := c_{i,t_i}.
\]

We define
\[
c^*
=
\max_{j \in r_i \cup \{i\}} c_{j \to i,t_i}^{\,\tau-1},
n_{j^*}
\in
\arg\max_{j \in r_i \cup \{i\}} c_{j \to i,t_i}^{\,\tau-1}.
\]

Agent \( n_i \) chooses its private social action \( h_i^\tau \in \mathcal{H}_i \),
which determines the diffused values
\( c_{t_i}^{i \to j,\,\tau} \) for each \( n_j \in r_i \). We analyze agent \( n_i \)'s diffusion incentives neighbor by neighbor.

\medskip
\noindent\textbf{\textit{Case 1: No diffusion to \( n_{j^*} \).}}

By definition of \( n_{j^*} \),
\[
c_{j^* \to i,t_i}^{\,\tau-1} = c^*.
\]
Hence agent \( n_{j^*} \) already observes a competence value weakly higher than any value agent
\( n_i \) can diffuse.
Therefore, under any feasible diffusion choice, agent \( n_{j^*} \) does not update its
delegation decision in response to \( h_i^\tau \).
Thus, in any best response,
\[
c_{i \to j^*,t_i}^{\,\tau} = \varnothing .
\]

\medskip
\noindent\textbf{\textit{{Case 2: Diffusion to non-maximal neighbors.}}}

We fix any \( n_j \in r_i \setminus \{n_{j^*}\} \),

\medskip
\noindent
\emph{Case 2.1: \( c_{j \to i,t_i}^{\,\tau-1} \ge c^* \).}

This case is impossible by definition of \( c^* \).

\medskip
\noindent
\emph{Case 2.2: \( c_{j \to i,t_i}^{\,\tau-1} < c^* \).}

If agent \( j \) is non-critical, changing \( c_{i \to j,t_i}^{\,\tau} \) does not affect agent \( n_i \)'s payment. If agent \( n_j \) is critical, failing to induce delegation from \( n_j \) implies that agent
\( n_i \) does not lie on the winning path, yielding zero payment. Hence, whenever \( n_j \) is critical, agent \( n_i \) must ensure
$
c_{i \to j,t_i}^{\,\tau} \ge c_{j \to i,t_i}^{\,\tau-1}
$ to secure \( n_j \)'s delegation. However, diffusing a value strictly larger than \( c^* \) cannot increase agent \( n_i \)'s
payment, since \( c^* \) is the maximal competence value available in \( n_i \)'s
neighborhood.
Therefore, any diffusion above \( c^* \) is weakly dominated. Thus, an optimal diffusion choice satisfies
\[
c_{i \to j,t_i}^{\,\tau}
=
\min\!\left(c_{j \to i,t_i}^{\,\tau-1} + \delta,\; c^* \right),
\]
where \( \delta > 0 \) is the minimal diffusion increment.

\medskip
Therefore, the diffusion policy
\[
c_{t_i}^{i \to j,\,\tau}
=
\begin{cases}
\varnothing, & n_j = n_{j^*}, \\[4pt]
\min\!\left(c_{j \to i,t_i}^{\,\tau-1} + \delta,\; c^* \right),
& n_j \in r_i \setminus \{n_{j^*}\} \text{ and } c_{j \to i,t_i}^{\,\tau-1} < c^*, \\[6pt]
\varnothing, & \text{otherwise},
\end{cases}
\]
constitutes a best response for agent \( n_i \) at round \( \tau \).
\hfill $\square$

\newpage
\section{\name{AgentSociety} Pseudocode}
\label{pseudo}

\begin{algorithm}[!ht]
\caption{AgentSociety: Mechanism for consensus based request routing among self-interested heterogeneous agents}
\label{alg:agentsociety_loop}
\SetKwInOut{Input}{Input}\SetKwInOut{Output}{Output}
\Input{Multi-task query $Q$, intrinsic competence $C_i$, initial diffusion state $\mathcal{S}_i^{\text{diff}} = \{ \hat{C}_{k \to i} \}_{k \in \mathcal{N}(i)}$}
\Output{Optimal path $\mathcal{P}^*$, Updated diffusion state $\mathcal{S}_i^{\text{diff}'}$}

\BlankLine
\tcp{Phase 1: Delegation and strategic competence reporting}
\ForEach{agent $i \in \mathcal{T}_{t_k}$ \textbf{in parallel}}{
    Observe current diffusion state $\mathcal{S}_i^{\text{diff}}$ (signals received from neighbors)\;
    Submit strategic action $a_i = (\mathbf{c}'_i, v_i, r_i)$ to the ledger where: \\
    $\mathbf{c}'_i = \max(C_i, \max_{k} \hat{C}_{k \to i})$ \tcp*{Signal derived from internal competence and received diffusion}
    $v_i = \arg\max_{k \in \mathcal{N}(i)} \hat{C}_{k \to i}$ \tcp*{Delegate to the most competent signaling neighbor}
}

\tcp{Phase 2: Consensus based request routing path}
\ForEach{task $t_j \in Q$}{
    Identify gurus $\mathcal{G}_j = \{g \in \mathcal{T}_{t_j} \mid v_g = g\}$\;
    Compute transitive vote counts $V(g) = \sum_{i: v_i \twoheadrightarrow g} 1$ via delegation graph traversal\;
}

Initialize feasible paths $\mathbb{F} = \emptyset$\;
\ForEach{$g \in \mathcal{G}_1$}{
    Select bridge $d^* = \arg\max_{d \in \mathcal{D}(g)} \sum_{j > 1} c_d^{t_j}$\;
    \If{bridge $d^*$ connects to guru $g' \in \mathcal{G}_{k+1}$}{
        $\mathbb{F} \leftarrow \mathbb{F} \cup \{\mathcal{P}_{g \to g'}\}$\;
    }
}
$\mathcal{P}^* = \arg\max_{\mathcal{P} \in \mathbb{F}} \sum_{g \in \mathcal{P}} V(g)$\;

\tcp{Phase 3: Payoff Calculation}
\ForEach{agent $v \in \mathcal{P}^*$}{
    \IF{$v$ is critical \COMMENT{via counterfactual audit}}{
        $P_v = p^{(t_j)}(v) - p^{\text{inf}}(v) - p^{\text{mis}}(v)$\;
    }
}

\tcp{Phase 4: Social Diffusion and State Update}
\ForEach{agent $i \in \mathcal{N}$ \textbf{in parallel}}{
    Observe realized reward $P_i$ and transition state $s_{t+1}$\;
    \eIf{$P_i < \text{Threshold}$ (Sub-optimal Outcome)}{
        \STATE \COMMENT{Modify diffusion action to garner future influence}
        $\hat{C}_{i \to j} \leftarrow C_i + \epsilon$ for $j \in \mathcal{N}(i)$\;
    }{
        $\hat{C}_{i \to j} \leftarrow C_i$ \tcp*{Maintain current signaling level}
    }
    Update Diffusion State $\mathcal{S}_i^{\text{diff}'} = \{ \hat{C}_{k \to i} \}_{k \in \mathcal{N}(i)}$ based on new signals from neighbors\;
    Update history $\Pi_i \leftarrow \Pi_i \cup \{P_i\}$\;
}
\end{algorithm}

\newpage

\section{Configurations and additional experimental results}

\subsection{Configuration for diffusion payment analysis Fig. \ref{fig:strategic_analysis} (right). }
\label{single_agent_config}

\begin{lstlisting}[
    language=Python,
    basicstyle=\ttfamily\footnotesize, % Smaller font for wide code
    frame=single,                      % Draws the box
    rulecolor=\color{black},           % Box border color
    breaklines=true,                   % Wraps long lines
    breakatwhitespace=true,            % Only wraps at spaces
    numbers=none,                      % Removes line numbers to save width
    xleftmargin=5pt,                   % Adds internal padding
    xrightmargin=5pt,
    showstringspaces=false             % Removes the visual 'u' symbols in strings
]
CUSTOM_GRAPH_CONFIG = {
    'nodes': [
        {'id': 1, 'primary_task': 1, 'intrinsic_competence': {1: 0.7, 2: 0, 3: 0}},
        {'id': 2, 'primary_task': 1, 'intrinsic_competence': {1: 0.7, 2: 0, 3: 0}},
        {'id': 3, 'primary_task': 1, 'intrinsic_competence': {1: 0.6, 2: 0, 3: 0}},
        {'id': 4, 'primary_task': 1, 'intrinsic_competence': {1: 0, 2: 0.4, 3: 0}},
        {'id': 5, 'primary_task': 1, 'intrinsic_competence': {1: 0, 2: 0.5, 3: 0}},
        {'id': 6, 'primary_task': 1, 'intrinsic_competence': {1: 0, 2: 0.4, 3: 0}},
        {'id': 7, 'primary_task': 1, 'intrinsic_competence': {1: 0, 2: 0.6, 3: 0}},
        {'id': 8, 'primary_task': 1, 'intrinsic_competence': {1: 0, 2: 0.6, 3: 0}},
        {'id': 9, 'primary_task': 1, 'intrinsic_competence': {1: 0, 2: 0.5, 3: 0}},
        {'id': 10, 'primary_task': 2, 'intrinsic_competence': {1: 0, 2: 0.9, 3: 0}},
        {'id': 11, 'primary_task': 2, 'intrinsic_competence': {1: 0, 2: 0.7, 3: 0}},
        {'id': 12, 'primary_task': 2, 'intrinsic_competence': {1: 0, 2: 0.6, 3: 0}},
        {'id': 13, 'primary_task': 2, 'intrinsic_competence': {1: 0, 2: 0.6, 3: 0}},
        {'id': 14, 'primary_task': 3, 'intrinsic_competence': {1: 0, 2: 0, 3: 0.5}}
    ],
    'edges': [
        (1, 4), (2, 6), (3, 7), (4, 8), (5, 6), (6, 11), (8, 9), 
        (9, 10), (10, 11), (11, 12), (11, 13), (8, 14), (11, 14), (7, 14)
    ]
}
\end{lstlisting}

\subsection{Configurations for LLM agent social intelligence characterization}
In all LLM agent experiments, Node 2 is assigned as the LLM agent. 
\label{confs}
\subsubsection{Configuration 1}
\label{conf1}
The dynamics for this configuration are presented in the main text in Fig. \ref{fig:combined_social_analysis}. 
\begin{lstlisting}[
    language=Python,
    basicstyle=\ttfamily\footnotesize, % Smaller font for wide code
    frame=single,                      % Draws the box
    rulecolor=\color{black},           % Box border color
    breaklines=true,                   % Wraps long lines
    breakatwhitespace=true,            % Only wraps at spaces
    numbers=none,                      % Removes line numbers to save width
    xleftmargin=5pt,                   % Adds internal padding
    xrightmargin=5pt,
    showstringspaces=false             % Removes the visual 'u' symbols in strings
]
CUSTOM_GRAPH_CONFIG_1 = {
    'nodes': [
        {'id': 1, 'intrinsic_competence': {1: 0.6}},
        {'id': 2, 'intrinsic_competence': {1: 0.3}},
        {'id': 3, 'intrinsic_competence': {1: 0.5}}
    ],
    'edges': [
        (1, 2),
        (2, 3),
    ]
}
\end{lstlisting}

We consider additional configurations to show that our mechanism is configuration-agnostic. 

\newpage

\subsubsection{Configuration 2 and its corresponding dynamics}
\label{conf2}

\begin{lstlisting}[
    language=Python,
    basicstyle=\ttfamily\footnotesize, % Smaller font for wide code
    frame=single,                      % Draws the box
    rulecolor=\color{black},           % Box border color
    breaklines=true,                   % Wraps long lines
    breakatwhitespace=true,            % Only wraps at spaces
    numbers=none,                      % Removes line numbers to save width
    xleftmargin=5pt,                   % Adds internal padding
    xrightmargin=5pt,
    showstringspaces=false             % Removes the visual 'u' symbols in strings
]
CUSTOM_GRAPH_CONFIG_2 = {
    'nodes': [
        {'id': 1, 'intrinsic_competence': {1: 0.6}},
        {'id': 2, 'intrinsic_competence': {1: 0.4}},
        {'id': 3, 'intrinsic_competence': {1: 0.5}},
        {'id': 4, 'intrinsic_competence': {1: 0.7}},
        {'id': 5, 'intrinsic_competence': {1: 0.3}},
        {'id': 6, 'intrinsic_competence': {1: 0.3}},
    ],
    'edges': [
        (1, 2),
        (2, 3),
        (3, 4),
        (2, 5),
        (2, 6)    
    ]
}
\end{lstlisting}

\label{sup_llm}
\begin{figure}[htbp] 
    \centering
    \includegraphics[width=0.5\columnwidth]{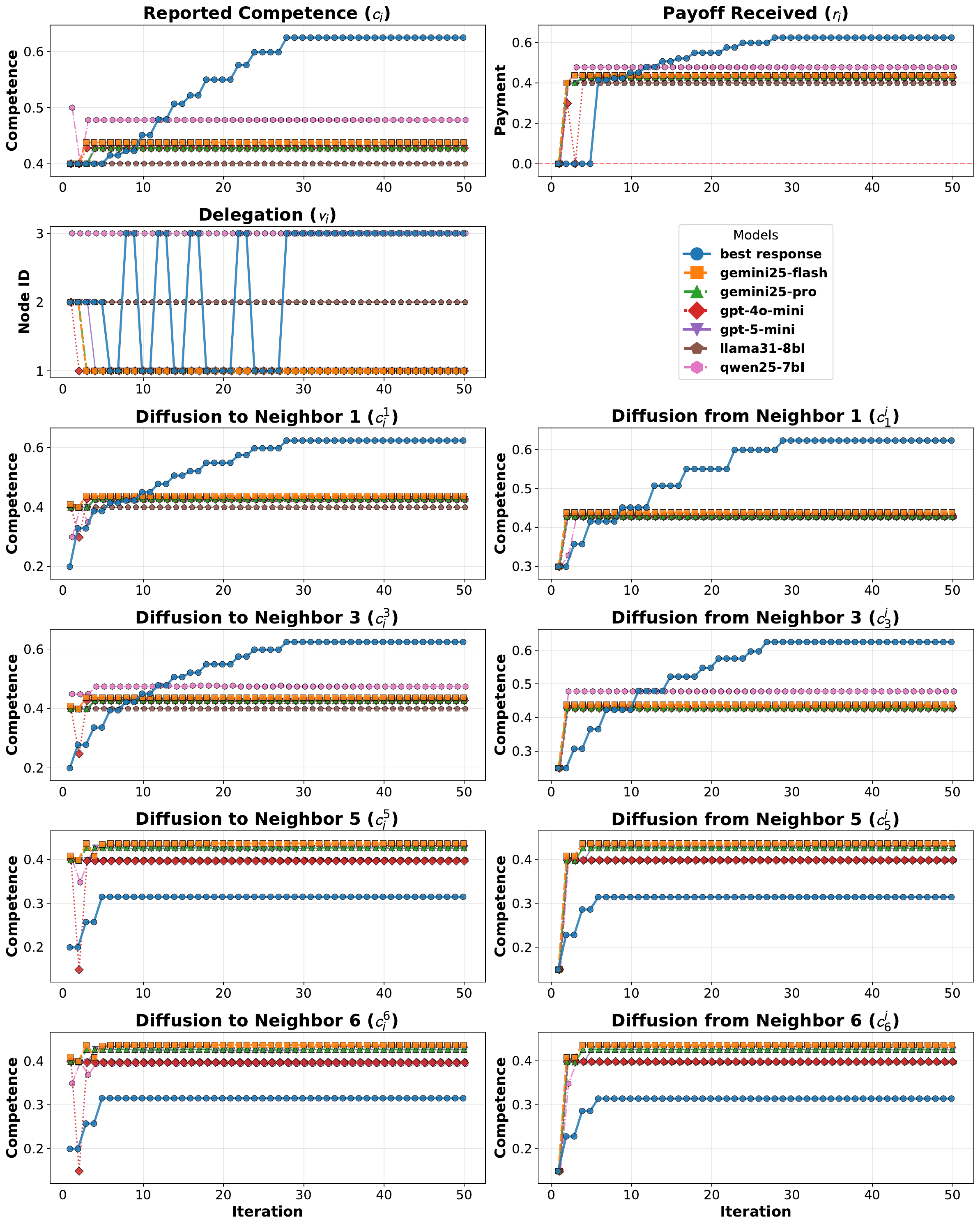}
    \caption{LLM Agent dynamics in \name{AgentSociety} benchmarked against best response for configuration 2}
    \label{fig:llm_plot_2}
\end{figure}

\subsubsection{Configuration 3 and its corresponding dynamics}
\label{conf3}

\begin{lstlisting}[
    language=Python,
    basicstyle=\ttfamily\footnotesize, % Smaller font for wide code
    frame=single,                      % Draws the box
    rulecolor=\color{black},           % Box border color
    breaklines=true,                   % Wraps long lines
    breakatwhitespace=true,            % Only wraps at spaces
    numbers=none,                      % Removes line numbers to save width
    xleftmargin=5pt,                   % Adds internal padding
    xrightmargin=5pt,
    showstringspaces=false             % Removes the visual 'u' symbols in strings
]
CUSTOM_GRAPH_CONFIG_3 = {
    'nodes': [
        {'id': 1, 'intrinsic_competence': {1: 0.7}},
        {'id': 4, 'intrinsic_competence': {1: 0.5}},
        {'id': 3, 'intrinsic_competence': {1: 0.4}},
        {'id': 2, 'intrinsic_competence': {1: 0.65}},
        {'id': 5, 'intrinsic_competence': {1: 0.4}},
        {'id': 6, 'intrinsic_competence': {1: 0.6}},
        {'id': 7, 'intrinsic_competence': {1: 0.5}},


    ],
    'edges': [
        (1, 4),
        (4, 3),
        (3, 2),
        (2, 5),
        (5, 6),
        (6, 7)
    ]
}
\end{lstlisting}

\begin{figure}[htbp] 
    \centering
    \includegraphics[width=0.6\columnwidth]{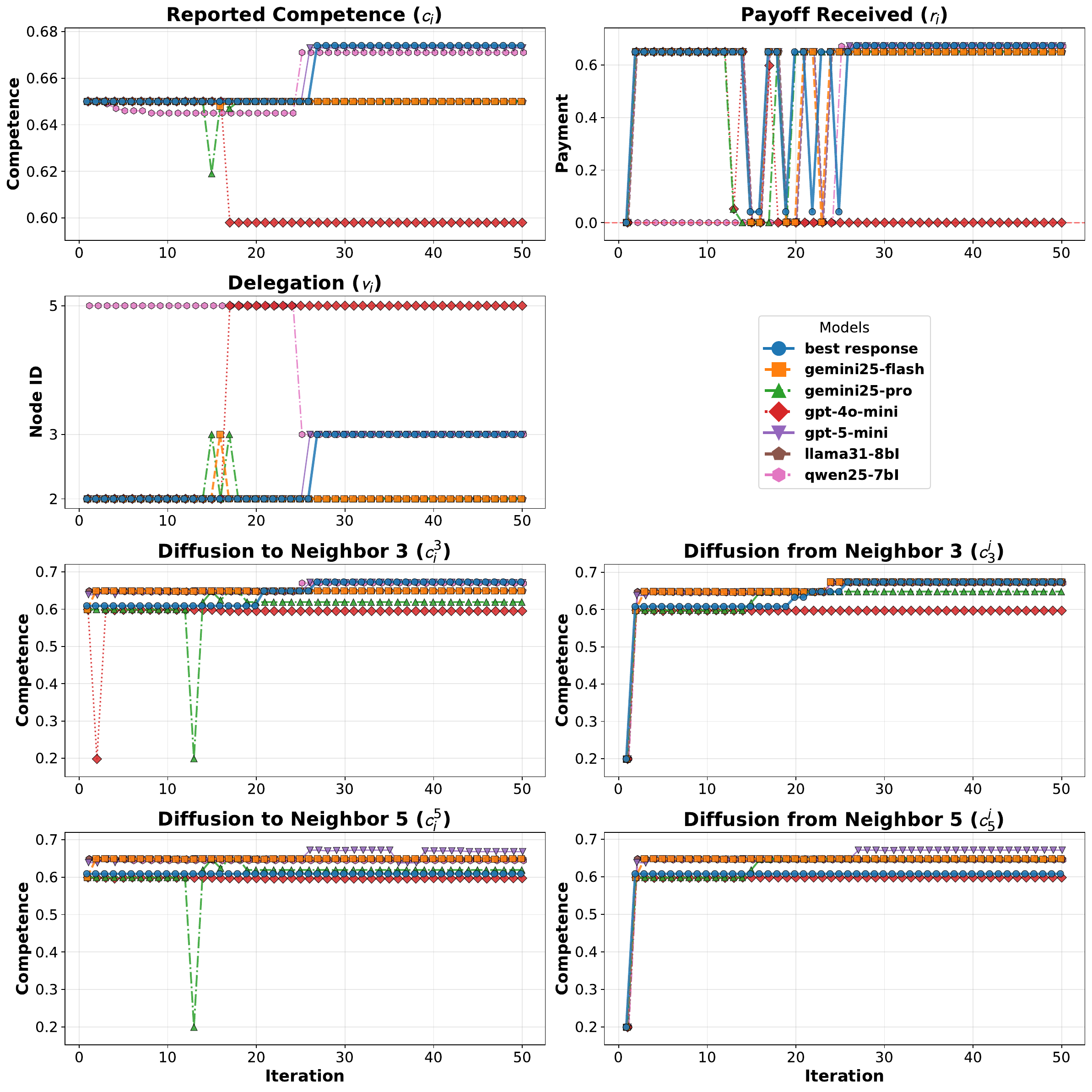}
    \caption{LLM Agent dynamics in \name{AgentSociety} benchmarked against best response for configuration 3}
    \label{fig:llm_plot_3}
\end{figure}

\section{Experimental Details}
We provide details on each of our experimental result here. All experiments were run on CPU with 16 GB RAM with LLM inferences accessed through API calls. Delegation and information diffusion decisions involve two additional calls to each LLM agent per task. In total, num-agents*2 additional calls per task represent the token overhead and cost, while the latency of the task is dictated by the maximum among the agents in the graph. As we leverage multiple heterogeneous LLMs, the latency and token costs vary by provider. Our average token overhead for delegation is ~1200 input tokens and ~20 output tokens while for diffusion the token overhead is ~1400 input tokens and ~45 output tokens, with costs for both in range of USD 0.0001 - USD 0.001 per call. Our highest latency is Gemini’s API driving about ~14s per call.

\subsection{Mechanism Characterization}
\label{mech}
We elaborate on the experiment details to generate Fig. \ref{fig:strategic_analysis} (left). We characterize the mechanism on a population of 17 nodes, where the topology is sampled randomly by partitioning agents into two tasks. Each agent adopts the best response strategy discussed in \S~\ref{intra_br}. Each agent's intrinsic competence $c$ is modeled using a clipped, squashed log-normal distribution ($z \sim \text{LogNormal}(-1.0, 1.0)$, $c = \text{clip}(z/(1+z), 0, 1)$) to reflect heavy-tailed expertise inequality, where a minority of agents possess high competence. The simulation executes for $T=50$ iterations. Each iteration follows a six-stage pipeline: (1) Delegation via the best response strategy; (2) Path Computation for root-to-guru chains; (3) Winner Determination based on aggregate votes; (4) Critical-node Analysis to identify essential agents; (5) Information Diffusion, where non-critical nodes diffuse competence values (6) Competence Update for future competence reports increment by $\Delta_r = 0.02$ subject to the neighborhood's diffused upper bound. To ensure statistical robustness, we perform 100 independent replications across different random graph seeds and report the mean and standard deviation for five key metrics. We track the Strategic Competence Report ($\bar{c}^{\text{rep}}_t$) to measure population-level reporting, User Realized Competence ($\bar{c}^{\text{real}}_t$) for the quality of service delivered by winning gurus, and Max Competence as the theoretical upper bound. Additionally, we log the Intrinsic Competence and the Payoff to Intermediaries ($\bar{p}_t$) to quantify the cost of delegation distributed to non-guru agents along the winning paths, and plot these in Fig. \ref{fig:strategic_analysis} (left).

\subsection{LLM Agent Social Intelligence in \name{AgentSociety}}
\label{si}
We elaborate on the experiment details to generate Fig. \ref{fig:combined_social_analysis} (left). We evaluate each LLM agent's social intelligence by introducing the LLM into \name{AgentSociety} with other nodes in the topography operating on the best response strategy. At each time step, the LLM agent must decide its delegation and diffusion behavior using natural language reasoning given the system context of the \name{AgentSociety} as well as prior history. Such an environment isolates the behavioral effects of replacing a rigid algorithmic participant with a free-form reasoner. We utilize a suite of static graph configurations (3–7 nodes) covering diverse topologies. At each iteration, the LLM agent is prompted twice—once for delegation and once for diffusion using their respective prompts, please see \S~\ref{app:prompts} for the prompts.

 We benchmark each of the 6 LLMs presented in Fig. \ref{fig:combined_social_analysis} (left) with temperature 0.0 (except gpt-5-mini that only takes 1.0) and a 4096 token limit. To ensure longitudinal coherence, the agent’s full per-iteration memory (previous delegation targets, payments received, and diffusion records) is fed back into the subsequent prompt, allowing the LLM to operate on its own historical context. The LLM agent observes payoffs given by Eq. \ref{eq:del_payoff} with $\alpha$ set to 100. Since all other nodes remain deterministic, any deviation from best response provides a clean ablation of the LLM's social intelligence impact on mechanism. We quantify LLM behavioral divergence from the best-response equilibrium across two granularities in the presence of other best response agents: decision-level divergence, representing the cumulative difference in actions given a fixed state, and trajectory-level divergence, representing the cumulative shift in the evolved state itself. These deviations are measured over multiple configurations through the percentage overlap in delegation choices and the mean absolute error (MAE) of diffused information and reported competence in Fig. \ref{fig:combined_social_analysis} (left).

\subsection{Collaborative Performance}
\label{cps}
In each of these experiments, we consider 6 models instantiated in \name{AgentSociety} with a fixed (randomly generated) graph topology (for uniformity). The capability vectors are of the dimension as domains present in the datasets and estimated using a 30\% train split and the evaluation is done on the remaining 70\%. 

\paragraph{MMLU-Pro} We utilize the TIGER-Lab/MMLU-Pro dataset, generating 5-shot predictions for each model via the official answer-extraction protocol and mapped to their 14 subdomains. The evaluation panel includes six diverse models: Llama-3.1-8B-Instruct, Qwen1.5-7B-Chat, Claude-3.5-Sonnet, Gemini-1.5-Pro-002, Gemini-2.0-Flash-Exp, and GPT-4o-Mini.

\begin{lstlisting}[
    language=Python,
    basicstyle=\ttfamily\footnotesize, % Smaller font for wide code
    frame=single,                      % Draws the box
    rulecolor=\color{black},           % Box border color
    breaklines=true,                   % Wraps long lines
    breakatwhitespace=true,            % Only wraps at spaces
    numbers=none,                      % Removes line numbers to save width
    xleftmargin=5pt,                   % Adds internal padding
    xrightmargin=5pt,
    showstringspaces=false             % Removes the visual 'u' symbols in strings
]
  Task IDs: 1=biology, 2=business, 3=chemistry, 4=computer science, 5=economics, 6=engineering, 7=health, 8=history,
  9=law, 10=math, 11=other, 12=philosophy, 13=physics, 14=psychology

  MMLU_PRO_CONFIG = {
      'nodes': [
          {'id': 1, 'intrinsic_competence': {1: 0.8951, 2: 0.7810, 3: 0.8208, 4: 0.8110, 5: 0.8160, 6: 0.6253, 7: 0.7676,
  8: 0.7434, 9: 0.5659, 10: 0.8574, 11: 0.7453, 12: 0.6884, 13: 0.8073, 14: 0.7806}},  # gemini-2.0-flash-exp
          {'id': 2, 'intrinsic_competence': {1: 0.8357, 2: 0.7111, 3: 0.6173, 4: 0.6951, 5: 0.7033, 6: 0.4109, 7: 0.7064,
  8: 0.5855, 9: 0.3682, 10: 0.7259, 11: 0.6775, 12: 0.5678, 13: 0.6262, 14: 0.7210}, },  # gpt-4o-mini
          {'id': 3, 'intrinsic_competence': {1: 0.8881, 2: 0.8032, 3: 0.6350, 4: 0.6951, 5: 0.7982, 6: 0.5943, 7: 0.7706,
  8: 0.7500, 9: 0.5591, 10: 0.5444, 11: 0.7127, 12: 0.7186, 13: 0.8189, 14: 0.8276}, },  # gemini-1.5-pro-002
          {'id': 4, 'intrinsic_competence': {1: 0.9196, 2: 0.7905, 3: 0.8031, 4: 0.8110, 5: 0.8309, 6: 0.6408, 7: 0.7829,
  8: 0.7697, 9: 0.6568, 10: 0.7537, 11: 0.7724, 12: 0.7236, 13: 0.7592, 14: 0.8056}, },  # claude-3.5-sonnet
          {'id': 5, 'intrinsic_competence': {1: 0.6783, 2: 0.4698, 3: 0.3761, 4: 0.4756, 5: 0.5341, 6: 0.2946, 7: 0.5352,
  8: 0.4605, 9: 0.2500, 10: 0.4593, 11: 0.4743, 12: 0.4874, 13: 0.3815, 14: 0.5893}, },  #  Meta-Llama-3_1-8B-Instruct
          {'id': 6, 'intrinsic_competence': {1: 0.5000, 2: 0.2540, 3: 0.1261, 4: 0.2744, 5: 0.4125, 6: 0.1628, 7: 0.2355,
  8: 0.3158, 9: 0.1773, 10: 0.2796, 11: 0.3388, 12: 0.2362, 13: 0.1599, 14: 0.4295}, },  # Qwen1.5-7B-Chat
      ],
      'edges': [(1, 2), (2, 3), (3, 4), (2, 5), (2, 6)],
  }
\end{lstlisting}

\paragraph{Open LLM Leaderboard v2} To prevent data duplication with MMLU-Pro above, we restrict this panel to the 15 non-MMLU sub-tasks, including IFEval, GPQA-Main, MATH (Algebra and Intermediate Algebra Hard), MuSR, and ten BBH sub-tasks. The model suite for this benchmark spans three capacity tiers: Llama-3.1-70B-Instruct, Tulu-3-70B, phi-4, Llama-3.1-8B-Instruct, Yi-1.5-34B-Chat, and Qwen2.5-Coder-7B-Inst.

\begin{lstlisting}[
    language=Python,
    basicstyle=\ttfamily\footnotesize, % Smaller font for wide code
    frame=single,                      % Draws the box
    rulecolor=\color{black},           % Box border color
    breaklines=true,                   % Wraps long lines
    breakatwhitespace=true,            % Only wraps at spaces
    numbers=none,                      % Removes line numbers to save width
    xleftmargin=5pt,                   % Adds internal padding
    xrightmargin=5pt,
    showstringspaces=false             % Removes the visual 'u' symbols in strings
]
Task IDs: 1=ifeval, 2=gpqa_main, 3=math_algebra_hard, 4=math_intermediate_algebra_hard, 5=musr_object_placements,
  6=bbh_boolean_expressions, 7=bbh_date_understanding, 8=bbh_disambiguation_qa, 9=bbh_formal_fallacies,
  10=bbh_geometric_shapes, 11=bbh_hyperbaton, 12=bbh_logical_deduction_five_objects,
  13=bbh_logical_deduction_seven_objects, 14=bbh_logical_deduction_three_objects, 15=bbh_movie_recommendation

  OPENLLM_V2_CONFIG = {
      'nodes': [
          {'id': 1, 'intrinsic_competence': {1: 0.8642, 2: 0.3209, 3: 0.0217, 4: 0.0000, 5: 0.2368, 6: 0.8933, 7: 0.5867,
  8: 0.7067, 9: 0.7867, 10: 0.2800, 11: 0.7467, 12: 0.5600, 13: 0.4800, 14: 0.9067, 15: 0.7733}, 'primary_task': 1},  # Llama-3.1-70B-Instruct
          {'id': 2, 'intrinsic_competence': {1: 0.5679, 2: 0.3433, 3: 0.4239, 4: 0.0714, 5: 0.3026, 6: 0.8933, 7: 0.5600,
  8: 0.7600, 9: 0.6267, 10: 0.4400, 11: 0.6933, 12: 0.5333, 13: 0.4400, 14: 0.7867, 15: 0.6667}, 'primary_task': 1},  # Yi-1.5-34B-Chat
          {'id': 3, 'intrinsic_competence': {1: 0.7037, 2: 0.3209, 3: 0.0000, 4: 0.0000, 5: 0.2105, 6: 0.8933, 7: 0.7200,
  8: 0.6800, 9: 0.7467, 10: 0.5200, 11: 0.8133, 12: 0.6400, 13: 0.5467, 14: 0.9200, 15: 0.6533}, 'primary_task': 1},  # phi-4
          {'id': 4, 'intrinsic_competence': {1: 0.6111, 2: 0.3134, 3: 0.0761, 4: 0.0000, 5: 0.2632, 6: 0.9333, 7: 0.3733,
  8: 0.5467, 9: 0.5867, 10: 0.5733, 11: 0.5333, 12: 0.4000, 13: 0.3067, 14: 0.7333, 15: 0.6267}, 'primary_task': 1},  # Qwen2.5-Coder-7B-Inst
          {'id': 5, 'intrinsic_competence': {1: 0.8457, 2: 0.3806, 3: 0.0000, 4: 0.0000, 5: 0.3289, 6: 0.8133, 7: 0.5600,
  8: 0.6533, 9: 0.6533, 10: 0.3333, 11: 0.6933, 12: 0.3600, 13: 0.3200, 14: 0.8267, 15: 0.8267}, 'primary_task': 1},  # Tulu-3-70B
          {'id': 6, 'intrinsic_competence': {1: 0.4691, 2: 0.3433, 3: 0.2826, 4: 0.0119, 5: 0.3421, 6: 0.7600, 7: 0.3733,
  8: 0.4800, 9: 0.4933, 10: 0.3600, 11: 0.6800, 12: 0.3867, 13: 0.3733, 14: 0.6400, 15: 0.6267}, 'primary_task': 1},  # Llama-3.1-8B-Instruct
      ],
      'edges': [(1, 2), (2, 3), (3, 4), (2, 5), (2, 6)],
  }
\end{lstlisting}

\paragraph{SWE-bench}. We leverage trajectories from princeton-nlp/SWE-bench-Verified test split. Evaluation domains are defined by the top-5 resolved repositories. We benchmark two distinct 6-agent configurations: a Strong/Generalist panel (comprising SWE-agent-Claude-4, Kimi-K2, Augment-Agent, Claude-3.5-Sonnet, DeepSeek-V3, SWE-agent-GPT-4o) and a Weaker/Complementary panel (including Enginelabs, Claude-3.5-Haiku, GRU, Codeshellagent-Gemini, SWERL-Llama3-70B, and Amazon-Nova-Premier).

\begin{lstlisting}[
    language=Python,
    basicstyle=\ttfamily\footnotesize, % Smaller font for wide code
    frame=single,                      % Draws the box
    rulecolor=\color{black},           % Box border color
    breaklines=true,                   % Wraps long lines
    breakatwhitespace=true,            % Only wraps at spaces
    numbers=none,                      % Removes line numbers to save width
    xleftmargin=5pt,                   % Adds internal padding
    xrightmargin=5pt,
    showstringspaces=false             % Removes the visual 'u' symbols in strings
]
  Task IDs: 1=django/django, 2=matplotlib/matplotlib, 3=scikit-learn/scikit-learn, 4=sphinx-doc/sphinx, 5=sympy/sympy

  SWEBENCH_STRONG_CONFIG = {
      'nodes': [
          {'id': 1, 'intrinsic_competence': {1: 0.6848, 2: 0.5385, 3: 0.6667, 4: 0.4118, 5: 0.7000}, },  # Kimi-K2
          {'id': 2, 'intrinsic_competence': {1: 0.5000, 2: 0.3846, 3: 0.6667, 4: 0.2941, 5: 0.5667}, }, # Claude-3.5-Sonnet
          {'id': 3, 'intrinsic_competence': {1: 0.6413, 2: 0.4615, 3: 0.7500, 4: 0.4706, 5: 0.6667}, }, # Augment-Agent
          {'id': 4, 'intrinsic_competence': {1: 0.7065, 2: 0.4615, 3: 0.6667, 4: 0.5882, 5: 0.6000}, }, # SWE-agent-Claude-4
          {'id': 5, 'intrinsic_competence': {1: 0.5000, 2: 0.4615, 3: 0.3333, 4: 0.3529, 5: 0.3000}, }, # DeepSeek-V3
          {'id': 6, 'intrinsic_competence': {1: 0.2609, 2: 0.0000, 3: 0.2500, 4: 0.0000, 5: 0.2667}, }, # SWE-agent-GPT-4o
      ],
      'edges': [(1, 2), (2, 3), (3, 4), (2, 5), (2, 6)],
  }
\end{lstlisting}

\begin{lstlisting}[
    language=Python,
    basicstyle=\ttfamily\footnotesize, % Smaller font for wide code
    frame=single,                      % Draws the box
    rulecolor=\color{black},           % Box border color
    breaklines=true,                   % Wraps long lines
    breakatwhitespace=true,            % Only wraps at spaces
    numbers=none,                      % Removes line numbers to save width
    xleftmargin=5pt,                   % Adds internal padding
    xrightmargin=5pt,
    showstringspaces=false             % Removes the visual 'u' symbols in strings
]
  Task IDs: 1=django/django, 2=matplotlib/matplotlib, 3=scikit-learn/scikit-learn, 4=sphinx-doc/sphinx, 5=sympy/sympy
  
   SWEBENCH_WEAK_CONFIG = {
      'nodes': [
          {'id': 1, 'intrinsic_competence': {1: 0.5000, 2: 0.5385, 3: 0.5833, 4: 0.4118, 5: 0.5000}, }, # Codeshellagent-Gemini
          {'id': 2, 'intrinsic_competence': {1: 0.4674, 2: 0.3846, 3: 0.7500, 4: 0.3529, 5: 0.4333}, }, # GRU
          {'id': 3, 'intrinsic_competence': {1: 0.4674, 2: 0.3077, 3: 0.8333, 4: 0.1765, 5: 0.5333}, }, # Claude-3.5-Haiku
          {'id': 4, 'intrinsic_competence': {1: 0.6413, 2: 0.3846, 3: 0.7500, 4: 0.0000, 5: 0.6667}, }, # Enginelabs
          {'id': 5, 'intrinsic_competence': {1: 0.4674, 2: 0.4615, 3: 0.5833, 4: 0.2941, 5: 0.3667}, }, # Amazon-Nova-Premier
          {'id': 6, 'intrinsic_competence': {1: 0.4239, 2: 0.4615, 3: 0.7500, 4: 0.0588, 5: 0.3667}, }, # SWERL-Llama3-70B
      ],
      'edges': [(1, 2), (2, 3), (3, 4), (2, 5), (2, 6)],
  }
\end{lstlisting}

\subsection{Multi-Task Collaborative Performance}
\label{cpm}
We evaluate multi-task delegation mechanism using evaluation traces from the open-llm-leaderboard/model-details dataset. The model suite spans three capacity tiers—including Large (Llama-3.1-70B-Instruct, Tulu-3-70B, phi-4), Mid (Yi-1.5-34B-Chat), and Small (Llama-3.1-8B-Instruct, Qwen2.5-Coder-7B-Inst) — evaluated across 5 reasoning domains from the IFEval, MATH, and BBH suites. To isolate the routing component's performance, we select these domains where no single model dominates, partitioning each into a 40\% train split for empirical competence estimation and a 60\% held-out test split. The experimental graph consists of 30 nodes (model-domain pairs) in a randomly generated topology. We benchmark multi-task combinations by running the delegation mechanism for 50 iterations with 1,000 sampled test questions per domain to ensure robust performance metrics. The \name{AgentSociety} mechanism runs similarly to \ref{mech}. 

\begin{lstlisting}[
    language=Python,
    basicstyle=\ttfamily\footnotesize, % Smaller font for wide code
    frame=single,                      % Draws the box
    rulecolor=\color{black},           % Box border color
    breaklines=true,                   % Wraps long lines
    breakatwhitespace=true,            % Only wraps at spaces
    numbers=none,                      % Removes line numbers to save width
    xleftmargin=5pt,                   % Adds internal padding
    xrightmargin=5pt,
    showstringspaces=false             % Removes the visual 'u' symbols in strings
]
  MULTI_TASK_CONFIG = {
      'nodes': [
          # ifeval (task=1)
          {'id': 1,  'intrinsic_competence': {1: 0.8611, 2: 0, 3: 0, 4: 0, 5: 0}, 'primary_task': 1},  #
  Llama-3.1-70B-Instruct
          {'id': 2,  'intrinsic_competence': {1: 0.8472, 2: 0, 3: 0, 4: 0, 5: 0}, 'primary_task': 1},  # Tulu-3-70B
          {'id': 3,  'intrinsic_competence': {1: 0.6991, 2: 0, 3: 0, 4: 0, 5: 0}, 'primary_task': 1},  # phi-4
          {'id': 4,  'intrinsic_competence': {1: 0.5741, 2: 0, 3: 0, 4: 0, 5: 0}, 'primary_task': 1},  # Yi-1.5-34B-Chat
          {'id': 5,  'intrinsic_competence': {1: 0.6111, 2: 0, 3: 0, 4: 0, 5: 0}, 'primary_task': 1},  # Qwen2.5-Coder-7B-Inst
          {'id': 6,  'intrinsic_competence': {1: 0.4769, 2: 0, 3: 0, 4: 0, 5: 0}, 'primary_task': 1},  # Llama-3.1-8B-Instruct
          # math_algebra_hard (task=2)
          {'id': 7,  'intrinsic_competence': {1: 0, 2: 0.0000, 3: 0, 4: 0, 5: 0}, 'primary_task': 2},  # Tulu-3-70B
          {'id': 8,  'intrinsic_competence': {1: 0, 2: 0.4180, 3: 0, 4: 0, 5: 0}, 'primary_task': 2},  # Yi-1.5-34B-Chat
          {'id': 9,  'intrinsic_competence': {1: 0, 2: 0.2869, 3: 0, 4: 0, 5: 0}, 'primary_task': 2},  # Llama-3.1-8B-Instruct
          {'id': 10, 'intrinsic_competence': {1: 0, 2: 0.0000, 3: 0, 4: 0, 5: 0}, 'primary_task': 2},  # phi-4
          {'id': 11, 'intrinsic_competence': {1: 0, 2: 0.0574, 3: 0, 4: 0, 5: 0}, 'primary_task': 2},  # Qwen2.5-Coder-7B-Inst
          {'id': 12, 'intrinsic_competence': {1: 0, 2: 0.0164, 3: 0, 4: 0, 5: 0}, 'primary_task': 2},  # Llama-3.1-70B-Instruct
          # bbh_geometric_shapes (task=3)
          {'id': 13, 'intrinsic_competence': {1: 0, 2: 0, 3: 0.5800, 4: 0, 5: 0}, 'primary_task': 3},  # Qwen2.5-Coder-7B-Inst
          {'id': 14, 'intrinsic_competence': {1: 0, 2: 0, 3: 0.4600, 4: 0, 5: 0}, 'primary_task': 3},  # Yi-1.5-34B-Chat
          {'id': 15, 'intrinsic_competence': {1: 0, 2: 0, 3: 0.4800, 4: 0, 5: 0}, 'primary_task': 3},  # phi-4
          {'id': 16, 'intrinsic_competence': {1: 0, 2: 0, 3: 0.3300, 4: 0, 5: 0}, 'primary_task': 3},  # Tulu-3-70B
          {'id': 17, 'intrinsic_competence': {1: 0, 2: 0, 3: 0.2800, 4: 0, 5: 0}, 'primary_task': 3},  # Llama-3.1-70B-Instruct
          {'id': 18, 'intrinsic_competence': {1: 0, 2: 0, 3: 0.3500, 4: 0, 5: 0}, 'primary_task': 3},  # Llama-3.1-8B-Instruct
          # bbh_disambiguation_qa (task=4)
          {'id': 19, 'intrinsic_competence': {1: 0, 2: 0, 3: 0, 4: 0.7900, 5: 0}, 'primary_task': 4},  # Yi-1.5-34B-Chat
          {'id': 20, 'intrinsic_competence': {1: 0, 2: 0, 3: 0, 4: 0.7000, 5: 0}, 'primary_task': 4},  # Llama-3.1-70B-Instruct
          {'id': 21, 'intrinsic_competence': {1: 0, 2: 0, 3: 0, 4: 0.6900, 5: 0}, 'primary_task': 4},  # phi-4
          {'id': 22, 'intrinsic_competence': {1: 0, 2: 0, 3: 0, 4: 0.6600, 5: 0}, 'primary_task': 4},  # Tulu-3-70B
          {'id': 23, 'intrinsic_competence': {1: 0, 2: 0, 3: 0, 4: 0.5800, 5: 0}, 'primary_task': 4},  # Qwen2.5-Coder-7B-Inst
          {'id': 24, 'intrinsic_competence': {1: 0, 2: 0, 3: 0, 4: 0.5000, 5: 0}, 'primary_task': 4},  # Llama-3.1-8B-Instruct
          # bbh_hyperbaton (task=5)
          {'id': 25, 'intrinsic_competence': {1: 0, 2: 0, 3: 0, 4: 0, 5: 0.8100}, 'primary_task': 5},  # phi-4
          {'id': 26, 'intrinsic_competence': {1: 0, 2: 0, 3: 0, 4: 0, 5: 0.7200}, 'primary_task': 5},  # Yi-1.5-34B-Chat
          {'id': 27, 'intrinsic_competence': {1: 0, 2: 0, 3: 0, 4: 0, 5: 0.7100}, 'primary_task': 5},  # Llama-3.1-70B-Instruct
          {'id': 28, 'intrinsic_competence': {1: 0, 2: 0, 3: 0, 4: 0, 5: 0.6700}, 'primary_task': 5},  # Tulu-3-70B
          {'id': 29, 'intrinsic_competence': {1: 0, 2: 0, 3: 0, 4: 0, 5: 0.6800}, 'primary_task': 5},  # Llama-3.1-8B-Instruct
          {'id': 30, 'intrinsic_competence': {1: 0, 2: 0, 3: 0, 4: 0, 5: 0.5200}, 'primary_task': 5},  # Qwen2.5-Coder-7B-Inst
      ],
      'edges': [
          (2, 3), (3, 4), (4, 5), (5, 6),
          (7, 8), (8, 9), (9, 10), (10, 11), (11, 12),
          (13, 14), (14, 15), (15, 16), (16, 17), (17, 18),
          (20, 21), (21, 22), (22, 23), (23, 24),
          (25, 26), (26, 27), (27, 28), (28, 29), (29, 30),
          (2, 7),  (6, 9),     
          (5, 13), (4, 14),  
          (4, 20), (2, 20),    
          (3, 25), (6, 29),   
          (10, 15), (8, 14),   
          (8, 20), (7, 22),    
          (10, 25), (9, 29),   
          (14, 20), (13, 23),  
          (15, 25), (18, 29),  
          (21, 25), (20, 26),  
      ],
  }
\end{lstlisting}

\lstset{
    basicstyle=\ttfamily\small,
    breaklines=true,               
    frame=single,                  
    backgroundcolor=\color{gray!5}, 
    columns=fullflexible,
    keepspaces=true
}

\newpage
\section{LLM Agent Prompts}
\label{app:prompts}

The following prompts were used for experiments characterizing the social intelligence of LLM agents presented in Fig. \ref{fig:combined_social_analysis}, with the same prompts used across all LLM agents. 

\begin{lstlisting}[caption={System Prompt (Describing AgentSociety)}]
You are a node in a liquid democracy network making strategic decisions about voting and information sharing in order to maximise your payment (reward). YOU ARE A UTILITY MAXIMISER WITH YOUR UTILITY = PAYMENT. The game is defined as follows

### Tasks and Competence
- There is a single task to be completed 
- Each node has ACTUAL COMPETENCE for the task (0.0 to 1.0 scale)
- Each node reports competence during delegaion based on which the payment (reward) received is computed. Each node is strategic here to report competence in such a way that the payment is maximized. 
- Each node also diffuses a competence to each of its neighbours (could be different to different neighbours). Each node is strategic here because higher than required diffusion hurts its payment share, while lower than required might not get it on the winning path which would mean that there is no payment altogether. 
- All agents maintain a log of what was diffused, so once a certain competence is diffused, not diffusing any information in the following iterations is the same as diffusing the previously diffused value as neighboring nodes log what was received from each node. 
- The goal of each node is to maximize its payment and not just receive a payment. 
- A positive payment means income and a negative payment means a penalty. Penalty is many times higher than income. 
- As the diffusion happens between neighbors each iteration, there is a delay for the diffused value to reach extended neighbors in the graph to then influence them to change their voting based on diffused value.
- Therefore the agents need to consider that the reward and penalty can sometimes be lagging with respect to the information diffused. 
- A rational agent starts with diffusing low values below its reported competence for payment and increases its diffusion value until required to maximize payment. 
- A rational agent does not report competence for payment to be higher than max(actual competence, competence obtained from neighbors) on the primary task. 

The system operates in the following steps - 
The system broadcasts the task to all nodes.
The first step for each participating node is to vote, either for themselves or for one of their neighbors.
After the voting process concludes, votes aggregate at gurus, defined as agents who vote for themselves. Each guru represents a group of agents that have voted to it, with voting occurring transitively. 
The voting paths are determined by the delegation paths formed by the transitive voting. The winning path is defined as the group with the highest number of agents.
Voting for a neighbor could be beneficial because there is finally only one winning path identified and only agents that make up this winning path have an opportunity to get paid.  The goal is to figure out who to vote for, given the payment mechanism in order to maximise own payment. 
During voting, each node also sends the reported competence, that it can deliver either by itself or leveraging its neighbors by voting to them, to the system that will eventually be used for payment calculation. 
Within this winning path, not all nodes are paid. Only the ones deemed critical, or those that would be part of a (different) winning path in case they had voted for themselves, are paid. These nodes make up the critical path sorted in the order of their competencies. 
After identifying the critical path the framework assigns payments to utilizing the formula below by plugging their respective reported competence into it. 
Let the critical path C_{t_k} be expressed as D_w^{t_k} = (a_1, a_2, ..., a_{k-1}, a_k = w), which denotes the ordered sequence of agents whose vote forwarding was necessary for w to receive the allocation of task t_k. The payment function (positive payment means the agent gets paid) for an agent a_i in D_w^{t_k} is then given by p^(s_{sigma_j})(a_i) equals:

*   I(c_{i+1} >= c_i) * (f(c_i) - f(c_{i-1})), if a_i is in D_w^{t_k} excluding {w};
*   I(c_{i+1} >= c_i) * (f(c_i) - f(c_{i-1})), if a_i = w;
*   0, if a_i is not in D_w^{t_k},

where: c_i is the competence reported of agent a_i on task t_k, c_{i-1} = 0 if a_i = a_1, c_{i+1} = true competence of winner if a_i = w. I(.) is the indicator function, returning +1 if the condition holds and alpha(negative value) otherwise, f(.) is a monotone function modulating payments, and c is a baseline cost to the winning agent w to execute the task. Intuitively, the payment policy rewards or penalizes intermediaries according to the incremental competence they contribute along the vote chain.
After the payment is made, all agents are allowed to make a decision based on their current information state and payment outcome on whether they want to diffuse information to their neighbors to improve their own utility or payment. 

\end{lstlisting}

\begin{lstlisting}[caption={Delegation Prompt}]
You are rational and intelligent Agent {node_id} in a system of all rational and intelligent agents making a delegation decision and reporting a competence for payment based on your delegation decision.

YOU NEED TO PAY ATTENTION AND UNDERSTAND TO THE PAYMENT FUNCTION IN ORDER TO MAXIMIZE YOUR PAYMENT for both your delegation decision as well as competence for payment. 

## Your Information
- Your Primary Task: Task {primary_task}
- Your Actual Competence: {intrinsic_competence}
- Tasks Being Voted On: {tasks}

## Your Memory (History from Previous Iterations)
{memory_info}

## Your Neighbors competence (same primary task only)
{neighbors_info}

## Decision Required
You must choose  for Task {primary_task}, who to delegate your vote. Your options are:
1. Vote for yourself (delegate to yourself)
2. Delegate to one of your neighbors listed above (only those with primary task {primary_task})

If you choose to vote for yourself, use your own node ID ({node_id}).

Based on who you delegated to, report a corresponding competence which determines your payment. 

Use your memory to learn from past delegation choices and their resulting payments. 
You need to ensure that you delegate to a neighbor that at least can guarantee the performance you have already diffused. Otherwise your neighbors will no longer trust you to be accurate in your claims.

## Response Format
Respond with ONLY a JSON object. ALL three fields are REQUIRED and **DO NOT PROVIDE AN EMPTY RESPONSE**:
{{"delegate_to": <node_id>, "competence_for_payment": {{<task_id>: <competence_value>}}, "reasoning": "brief explanation"}}

Example response:
{{"delegate_to": 5, "competence_for_payment": {{"1": 0.7}}, "reasoning": "Node 5 has ..."}}

Field descriptions:
- delegate_to: The node ID you are delegating to (use your own ID {node_id} if voting for yourself)
- competence_for_payment: REQUIRED. A dictionary mapping each task_id (integer) to the competence value (float) you want to report for payment calculation. Include ALL tasks from {tasks}. 
- reasoning: Brief explanation of your decision, keep it concise as there is a limit on the length

\end{lstlisting}

\begin{lstlisting}[caption={Diffusion Prompt}]

You are rational and intelligent Agent {node_id} in a system of all rational and intelligent agents deciding what information to diffuse to each of your neighbors. 
Bear in mind that diffusion at this iteration can depend on the delegation decision of the same iteration. 

## Your Information
- Your Primary Task: Task {primary_task}
- Your Actual Competence: {intrinsic_competence}
- Delegation Decision of current iteration: {delegated_to_info}
- Competence Reported for Payment (current iteration): {competence_for_payment_info}

## Your Memory (History from Previous Iterations)
{memory_info}

## Payment Outcome
{payment_status}

## Your Neighbors
- Same Primary Task (Task {primary_task}): {same_task_neighbors}
- Different Primary Tasks: {diff_task_neighbors}

## Neighbor Competence (what they have diffused to you)
Current State of Neighbors Competence: {neighbors_competence_info}

- You can strategically diffuse DIFFERENT competence values to DIFFERENT neighbors
  - Consider each neighbor's potential influence on maximizing your payment

- Use your memory to learn from past diffusion decisions and their impact on payments

## For example
If this is the received information -

# Your Information
- Your Primary Task: Task 1
- Your Actual Competence: {{"1": 0.2}}
- Delegation Decision of current iteration: You delegated to Node 5
- Competence Reported for Payment (current iteration): {{"1": 0.439}}

# Current State of Neighbors Competence: - Node 3 (Primary Task 1): Competence = {{"1": 0.419}}
- Node 5 (Primary Task 1): Competence = {{"1": 0.439}}

Here, diffusing a higher value of 0.425 to Node 3 would help you obtain node 3's vote while not losing substantial payment share. In this case, you would be acting as an intermediary transferring information from one node to another. 


## Decision Required
Decide whether to diffuse information and what competence values to diffuse to EACH neighbor individually.

## Response Format
Respond with ONLY a JSON object in this exact format:
{{
  "should_diffuse": true/false,
  "neighbor_updates": {{
    "<neighbor_id>": {{"<task_id>": competence_value, ...}},
    "<neighbor_id>": {{"<task_id>": competence_value, ...}}
  }},
  "reasoning": "brief explanation of strategy"
}}

Example:
{{
  "should_diffuse": true,
  "neighbor_updates": {{
    "2": {{"1": 0.7}},
    "5": {{"1": 0.5}}
  }},
  "reasoning": "Reporting higher competence to node 2 because of ..."
}}

If should_diffuse is false, the updates will be ignored.
You must provide an update for each neighbor you want to diffuse information to.
\end{lstlisting}

\end{document}